\documentclass{article}
\usepackage{fullpage}
\usepackage{url}
\usepackage{epsfig}
\usepackage{graphicx}
\usepackage{cite}
\usepackage{amsthm}
\usepackage[cmex10]{amsmath}
\usepackage{amssymb}
\usepackage[usenames,dvipsnames]{color}
\usepackage{algorithmic}
\usepackage[ruled,lined,noend,linesnumbered,commentsnumbered]{algorithm2e}
\usepackage{multirow}
\usepackage{epstopdf}
\usepackage{array}
\usepackage{eqparbox}
\usepackage{subfigure}
\usepackage{caption}

\newcommand{\eat}[1]{}

\newtheorem{definition}{Definition}

\newtheorem{theorem}{Theorem}

\def\0{{\mathbf 0}}
\def\1{{\mathbf 1}}

\DeclareMathOperator*{\argmax}{arg\,max}

\newcommand{\X}{\mathcal{X}}

\newcommand{\U}{\mathcal{U}}
\newcommand{\N}{\mathcal{N}}

\newcommand{\ra}{\mbox{$\rightarrow$}}

\newcommand{\reals}{\mathbb{R}}

\newcommand{\SPhard}{\textbf{\#P}-hard}
\newcommand{\NPhard}{\textbf{NP}-hard}

\newcommand{\promax}{{\sc ProMax}}
\newcommand{\infmax}{\textsc{InfMax}}

\newcommand{\Nin}{\mathrm{N}^\mathit{in}} 
\newcommand{\vals}{\mathbf{v}}

\newcommand{\p}{\mathbf{p}}

\newcommand{\greedy}{\textsf{Greedy}}
\newcommand{\UG}{\textsf{U-Greedy}}
\newcommand{\allomp}{\textsf{All-OMP}}
\newcommand{\ffs}{\textsf{FFS}}
\newcommand{\page}{\textsf{PAGE}}
\newcommand{\erf}{\mathrm{erf}}
\newcommand{\pihat}{\hat{\pi}}

\begin{document}

\title{Profit Maximization over Social Networks\footnote{An abbreviated version of this paper appears in the {\em Proceedings of the 12th IEEE International Conference on Data Mining (ICDM 2012)}, Brussels, Belgium, December 10 -- 13, 2012. The copyright of the conference version belongs to IEEE.}}

\author{
\begin{tabular}{cc}
 	Wei Lu & Laks V.S. Lakshmanan \\
	Dept. of Computer Science & Dept. of Computer Science \\
	University of British Columbia & University of British Columbia \\
	Vancouver, B.C., Canada & Vancouver, B.C., Canada \\
	{\tt welu@cs.ubc.ca} & {\tt laks@cs.ubc.ca} \\
\end{tabular}
}

\maketitle

\begin{abstract}
Influence maximization is the problem of finding a set of influential users in a social network such that the expected spread of influence under a certain propagation model is maximized. 
Much of the previous work has neglected the important distinction between social influence and actual product adoption.
However, as recognized in the management science literature, an individual who gets influenced by social acquaintances may not necessarily adopt a product (or technology), due, e.g., to monetary concerns.
In this work, we distinguish between influence and adoption by explicitly modeling the states of being influenced and of adopting a product. 
We extend the classical Linear Threshold (LT) model to incorporate prices and valuations, and factor them into users' decision-making process of adopting a product.
We show that the expected profit function under our proposed model maintains submodularity under certain conditions, but no longer exhibits monotonicity, unlike the expected influence spread function. 
To maximize the expected profit under our extended LT model, we employ an unbudgeted greedy framework to propose three profit maximization  algorithms.
The results of our detailed experimental study on three real-world datasets demonstrate that of the three algorithms, \textsf{PAGE}, which assigns prices dynamically based on the profit potential of each candidate seed, has the best performance both in the expected profit achieved and in running time.
\end{abstract}


\section{Introduction}\label{sec:intro}
The rapidly increasing popularity of online social networking
	sites such as Facebook, Google+, and Twitter has facilitated immense
	opportunities for large-scale viral marketing.
Viral marketing was first introduced to the data mining community
	by Domingos and Richardson~\cite{domingos01, richardson02};
it is a cost-effective method to promote a new product
	(or technology) by giving free or discounted samples 
	to a selected group of influential individuals,
	in the hope that through the word-of-mouth effects over the social
	network, a large number of product adoptions will occur.

Motivated by viral marketing, \emph{influence maximization} (\infmax{})
	has emerged as a fundamental problem concerning the propagation
	of innovations through social networks.
In their seminal paper, Kempe et al.~\cite{kempe03} formulated
	\infmax\ as a discrete optimization problem: 
given a directed graph $G=(V,E)$ (representing a social network)
	with pairwise influence weights on edges, and a positive number $k$,
	find $k$ individuals or seeds, such that by activating them initially, the
	{\em expected spread of influence} (or \emph{spread} for short)
	in the network under a certain propagation model is maximized.
Two classical propagation models studied in \cite{kempe03} are the
	{\em Independent Cascade} (IC) and the {\em Linear Threshold}
	(LT) model.
In this paper, we focus on the LT model, the details of which are
	deferred to Sec.~\ref{sec:bgd}.

The expected spread of influence of any set $S\subseteq V$, denoted by $h(S)$\footnote{The standard notation for the influence function is $\sigma$~\cite{kempe03}, but since $\sigma$ is used for the normal distribution $\N(\mu,\sigma^2)$ in the paper, we use $h$ here.}, is defined as the expected number
	of activated nodes after the diffusion process starting from $S$ quiesces.
Under both IC and LT models, \infmax{} is \NPhard{} and the influence
	function $h$ is monotone and submodular.
A set function $f\colon 2^{X}\ra\, \reals$ is {\em monotone} if $f(S) \leq f(T)$ whenever $S \subseteq T \subseteq X$, where $X$ is the ground set.
The function $f$ is {\em submodular} if $f(S \cup \{x\}) - f(S) \geq f(T \cup \{x\}) - f(T)$ holds for all $S \subseteq T \subseteq X$ and $x \in X \setminus T$.
Submodularity naturally captures the law of diminishing marginal returns,
	a fundamental principle in microeconomics.
With these good properties, approximation guarantees can be provided
	for \infmax{}~\cite{kempe03}.

Although \infmax{} has been studied extensively, a majority of the previous work has 
	focused on the classical propagation models, namely IC and LT, which do not fully incorporate
	important {\em monetary} aspects in people's decision-making process
	of adopting new products. The importance of such aspects is seen in
	actual scenarios and recognized in the management
	science literature. 
As real-world examples, until recently, Apple's iPhone has seemingly created bigger buzz 
	in social media than any other smartphones.
However, its worldwide market
	share in 2011 fell behind Nokia, Samsung, and LG\footnote{IDC Worldwide
	Mobile Phone Tracker, July 28, 2011.}.
This is partly due to the fact that iPhone is pricier both in hardware (if one buys it
	contract-free and factory-unlocked) and in its monthly rate plans.
On the contrary, the HP TouchPad was shown little interest by the tablet market
	when it was initially priced at $\$499$ ($16$GB).
However, it was sold out
	within a few days after HP dropped the price substantially to $\$99$ ($16$GB)~\footnote{\url{http://www.pcworld.com/article/237088/hp_drops_touchpad_price_to_spur_sales.html}}.  

In management science, the adoption of a new product is
	characterized as a two-step process~\cite{kalish85}.
In the first step, ``{\em awareness}'', an individual gets exposed to
	the product and becomes familiar with its features.
In the second step, ``{\em actual adoption}'', a person who is aware of the product will purchase it if her valuation outweighs the price.
Product awareness is modeled as being propagated through the word-of-mouth
	of existing adopters, which is indeed articulated by classical propagation
	models. 
However, the actual adoption step is not captured in these 
	classical models and is indeed the gap
	between these models and that in~\cite{kalish85}.

In a real marketing scenario, viral or otherwise, products are priced and
	people have their own valuations for owning
	them, both of which are critical in making adoption decisions. 
Precisely, the {\em valuation} of a person for a certain product is
	the maximum money she is willing to pay for it;
	the valuation for not adopting is defined to be zero~\cite{klbBook}.
Thus, when a company attempts to maximize its expected profit in a
	viral marketing campaign, such monetary factors 
	need to be taken into account.
However, in \infmax{}, only influence weights
	and network structures are considered, and the marketing strategies
	are restricted to binary decisions: for any node in the network, an
	\infmax{} algorithm just decides whether or not it should be seeded.

To address the aforementioned limitations, we propose the problem of \emph{profit
	maximization} (\promax{}) 
	over social networks, by incorporating both prices and valuations.
\promax{} is  the problem of finding an optimal strategy to maximize the expected
	total profit earned by the end of an influence diffusion process under
	 a given propagation model.
We extend the LT model to propose a new propagation
	model named the {\em Linear Threshold model with user
	Valuations} (LT-V), which explicitly introduces the states
	{\em influenced} and {\em adopting}.
Every user will be quoted a price by the company, and an {\em influenced}
	user  adopts, i.e., transitions to {\em adopting},
	only if the price does not exceed her valuation.

As pointed out by Kleinberg and Leighton~\cite{bobkleinberg03},
	people typically do not want to reveal their valuations before the price
	is quoted for reasons of trust.
Moreover, for privacy concerns, after a price is quoted, they usually
	only reveal their decision of adoption (i.e., ``yes'' or ``no''),
	but do not wish to share information about their true valuations.
Thus, following the literature~\cite{bobkleinberg03,klbBook},
	we make the {\em independent private value} (IPV) assumption,
	under which the valuation of each user is drawn independently
	at random from a certain distribution. 
Such distributions can be learned by a marketing company from historical sales data.
Furthermore, our model assumes users to be {\em price-takers}
	who respond myopically to the prices offered to them, solely based on
	their privately-held valuations and the price offered.

Since prices and valuations are considered in the optimization,
	marketing strategies for \promax{} require non-binary decisions:
	for any node in the network, we (i.e., the system) need to decide whether or
	not to seed it, and what price should be quoted.
Given this factor, the objective function to optimize in \promax, i.e,
	the expected total profit, is a function of both the seed set
	and the vector of prices.
As we will show in Secs.~\ref{sec:model} and \ref{sec:algo}, since discounting
	may be necessary for seeds, the profit function is
	in general {\em non-monotone}.
Also, as we show, the profit function maintains submodularity {\sl for any fixed vector of prices},
	regardless of the specific forms of valuation distributions. 

In light of the above, \promax{} is inherently more complex than \infmax{},
	and calls for more sophisticated algorithms for its solution.
As the profit function is in the form of the difference between a monotone submodular set function and a linear function,
	we first design an ``unbudgeted'' greedy (\UG{}) framework for seed set selection.
In each iteration, it picks the node with the largest expected marginal
	profit until the total profit starts to decline.
We show that for any fixed price vector, \UG{} provides quality
	guarantees slightly lower than a $(1-1/e)$-approximation.
To obtain complete profit maximization algorithms,
	we propose to integrate \UG{} with three pricing strategies, which leads to
	three algorithms \allomp{} (Optimal Myopic Prices), \ffs{}
	(Free-For-Seeds), and \page{} (Price-Aware GrEedy).
The first two are baselines and choose prices in ad hoc ways, 
while \page{}
	dynamically determines the optimal price to be offered to each candidate seed
	in each round of \UG{}.
Our experimental results on three real-world network datasets illustrate that
	\page{} outperforms \allomp{} and \ffs{} in terms of expected
	profit achieved and running time, and is more robust against various
	network structures and valuation distributions.

\paragraph{Road-map.}
Sec.~\ref{sec:bgd} discusses related work.
Sec.~\ref{sec:model} describes the LT-V model and defines \promax{}.
Sec.~\ref{sec:algo} presents our profit maximization algorithms.
We discuss experiments in Sec.~\ref{sec:exp}, and present extensions and conclusions
	in Sec.~\ref{sec:discuss}.

\section{Background and Related Work}\label{sec:bgd}
Domingos and Richardson~\cite{domingos01,richardson02}
	first posed \infmax{} as a data mining problem.
They modeled the problem using Markov random fields and proposed
	heuristic solutions.
Kempe et al.~\cite{kempe03} studied \infmax{} as a discrete
	optimization problem, and utilized submodularity
	of the spread function $h$  to obtain a greedy $(1-1/e)$-approximation algorithm using the results in 
\cite{submodular} (see Algorithm~\ref{alg:greedy1}: {\sf Greedy}). 
\greedy{} starts from an empty set; in each iteration it 
	adds to $S$ the element with the largest marginal gain 
	until $|S| = k$.

\paragraph{The Linear Threshold Model.}
We now describe the LT model~\cite{kempe03} in detail.
In this model, each node $u_i$ chooses an activation threshold
	$\theta_i$ uniformly at random from $[0,1]$,
	representing the minimum weighted fraction of active in-neighbors
	necessary so as to activate $u_i$.
Each edge $(u_i, u_j)\in E$ is associated with an influence weight
	$w_{i,j}$; for each $u_j\in V$, $\sum_{u_i\in\Nin(u_j)} w_{i,j} \le 1$,
	where $\Nin(u_i)$ is the set of in-neighbors of $u_i$ (i.e.,
	the sum of incoming weights does not exceed $1$). 
Time proceeds in discrete steps.
At time $0$, a seed set $S$ is activated.
At any time $t \ge 1$, we activate any inactive $u_i$
	if the total influence weight from its active in-neighbors
	reaches or exceeds $\theta_i$. 
Once a node is activated, it stays active.
The diffusion process completes when no more nodes can be activated.

Chen et al.~\cite{ChenYZ10} showed that it is \SPhard{} to compute the exact expected spread of any node set
	in general graphs for the LT model.
Thus, a common practice is to estimate the spread using
	Monte-Carlo (MC) simulations, in which case the approximation
	ratio of \greedy{} drops to $1-1/e-\epsilon$, where $\epsilon > 0$ depends on the number of MC simulations run~\cite{kempe03}.
By further exploiting submodularity,
	\cite{Leskovec07} proposed the cost-effective
	lazy forward (CELF) optimization, which improves the running
	time of \greedy{} by up to $700$ times.

Recently, Bhagat et al.~\cite{smriti12} addressed the difference between product
	adoption and influence in their LT-C (Linear Threshold with Colors) model.
In LT-C, the extent to which a node is influenced by
	its neighbors depends on two factors:
	influence weights and the opinions of the neighbors.
LT-C also features a ``tattle'' state for nodes:
	if an influenced node does not adopt, it 
	may still propagate positive or negative influence to neighbors.
However, unlike us, the LT-C model does not consider monetary aspects
	in product adoption.

Considerable work has been done on pricing in social networks.
Hartline et al.~\cite{Hartline08} studied
	optimal marketing for digital goods in social networks
	and proposed the influence-and-exploit (IE) framework.
In IE, seeds are offered free samples, and the seller can approach other users in
	a random sequence, bypassing the network structure.
Arthur et al.~\cite{arthur09} adopted IE
	to study a similar problem in which users arrive in a sequence
	decided by a cascade model (IC).
However, in~\cite{arthur09}, seeds are given as input (with free samples offered), whereas
	in our case, the choice of the seed set and of the prices
	are driven by profit maximization. These choices are made 
	by the algorithms. 
Work in \cite{bloch08,zeyuan11} formulated
	pricing in social networks as simultaneous-move games
	and studied equilibria of the games, whereas we focus on
	stochastic propagation models with social influence.

{
	\begin{algorithm}[t!]
	\caption{\greedy{} ($G=(V,E)$, $k$, $h$)}\label{alg:greedy1}
	$S \gets \emptyset$\;
	\For {$i = 1 \rightarrow k$} {
		$u \gets \argmax_{u_i\in V\setminus S} \left[ h(S\cup\{u_i\}) - h(S) \right]$\;
		$S \gets S \cup \{u\}$\;
	}
	Output $S$\;
	\end{algorithm}
}

\section{Linear Threshold Model with User Valuations and Its Properties}\label{sec:model}
In Sec.~\ref{sec:modeldef}, we describe our proposed LT-V model and define the profit maximization problem (\promax{}).
We then study a restricted case of \promax{} and present theoretical results for it in Sec.~\ref{sec:basic}.
In Sec.~\ref{sec:property}, we establish the submodularity result for general \promax{}.

\subsection{Model and Problem Definition}\label{sec:modeldef}
In the LT-V model, the social network is modeled as a
	directed graph $G=(V,E)$, in which each node
	$u_i\in V$ is associated with a \emph{valuation} $v_i\in [0,1]$.
Recall that in Sec.~\ref{sec:intro}, we made the IPV assumption
	under which valuations are drawn independently at random from some
	continuous probability distribution assumed known to the marketing company.
Let $F_i(x) = \Pr[v_i \le x]$ be the distribution function of $v_i$, and $f_i(x) = \frac{\mathrm{d}}{\mathrm{d}x}F_i(x)$ be the corresponding density function.
The domain of both functions is $[0,1]$ as we assume both prices and valuations are in $[0,1]$.
As in the classical LT model, each node $u_i$ has an influence threshold
	$\theta_i$ chosen uniformly at random from
	$[0,1]$. Each edge $(u_i,u_j)\in E$ has an influence
	weight $w_{i,j} \in [0,1]$, such that for each node $u_j$, $\sum_{u_i\in\Nin(u_j)} w_{i,j} \le 1$.
If $(u_i,u_j)\not\in E$, define $w_{i,j} = 0$.
Following~\cite{domingos01,richardson02}, we assume that there
	is a constant acquisition cost $c_a\in [0,1)$ for marketing to each seed
	(e.g., rebates, or costs of mailing ads
	and coupons).

\begin{figure}
	\centering
	\includegraphics[width=0.6\textwidth]{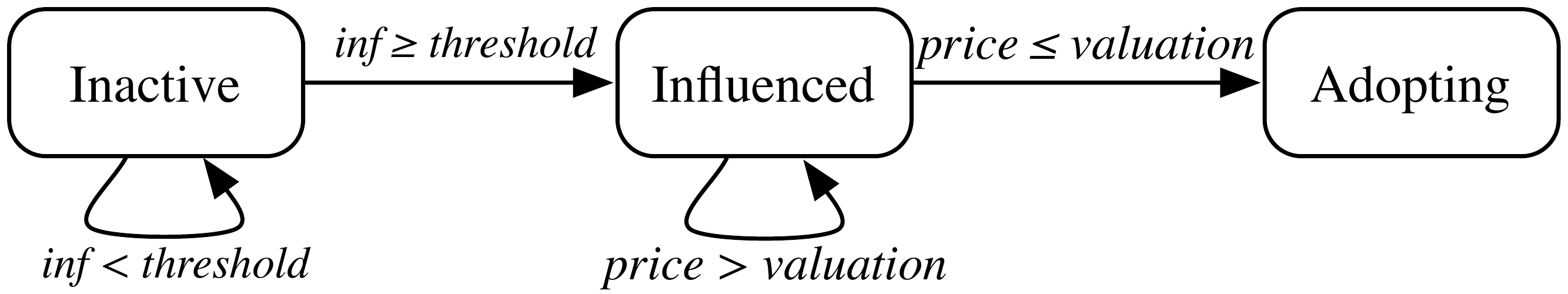}
	\caption{Node states in the LT-V model.}
	\label{fig:states}
\end{figure}

\paragraph{Diffusion Dynamics.}
Fig.~\ref{fig:states} presents a state diagram for the LT-V model.
At any time step, nodes are in one of the three states:
	\emph{inactive, influenced}, and {\em adopting}.
A diffusion under the LT-V model proceeds in discrete time steps.
Initially, all nodes are inactive.
At time $0$, a seed set $S$ is targeted and becomes influenced. Next, 
every user $u_i$ in the network is offered a price $p_i$ by the system. 
Let $\p = (p_1,\dotsc,p_{|V|})\in [0,1]^{|V|}$ denote a vector of quoted prices, which remains constant throughout the diffusion.
For any $u_i\in S$, it gets one chance to adopt (enters {\em adopting} state) at step $0$
	if  $p_i \le v_i$; otherwise it stays \emph{influenced}.

At any time $t\ge 1$, an inactive node $u_j$ becomes
	influenced if the total influence from its {\sl adopting
	in-neighbors} reaches its threshold, i.e.,
	$\sum_{u_i \in \Nin(u_j) \& \, u_i\, \text{adopting}} w_{i,j} \ge \theta_j.$
Then, $u_j$ will transition to \emph{adopting} at $t$ if $p_j\le v_j$,
	and will stay \emph{influenced} otherwise.
The model is progressive, meaning that all adopting nodes remain
	as adopters and no influenced node will ever become inactive. 
The diffusion ends if no more nodes can change states.

Following~\cite{kalish85}, we assume that only {\em adopting}
	nodes propagate influence, as adopters can release
	experience-related product features (e.g., durability, usability),
	making their recommendations more effective in removing
	doubts of inactive users.
This distinguishes our model from LT-C \cite{smriti12}, and in fact, the extensions to the LT model employed in LT-C and in LT-V are orthogonal and address different aspects in propagations of influence and adoption.

Formally, we define $\pi \colon 2^V \times [0,1]^{|V|} \to \reals$ to be the
	{\em profit function} such that $\pi(S,\p)$ is the expected
	(total) profit earned by the end of a diffusion process under the
	LT-V model, with $S$ as the seed set and $\p$ as the vector of prices.
The problem studied in the paper is as follows. 

\begin{definition}[Profit maximization (\promax)]\label{def:promax}
Given an instance of the LT-V model consisting of a graph $G = (V, E)$ with edge weights, 
find the optimal pair of a seed set $S$ and a price vector $\p$ that maximizes the 
expected profit $\pi(S, \p)$. 
\end{definition} 

\paragraph{The Virtual Mechanism and Its Truthfulness Guarantee.}
Recall that users are assumed to be price-takers making
	adoption decisions just by comparing the quoted price
	to their valuation.
Thus, it is natural to ask: \textsl{would an influenced user be better off
	by acting strategically, i.e., by {\em not} {deciding solely by}  
	comparing her true valuation to the price?}
In other words, for any pricing strategy used by a company,
	is it robust against strategic behaviors of users?

In fact, the price-taking procedure in LT-V can be structured as
	a {\em virtual} mechanism that we show is {\em truthful},
	and hence the dominant strategy for all users is to use true valuation.
It is worth emphasizing that the mechanism is {\em virtual} since in our model, the
	company needs {\em not} to run it and users will {\em not} 
	be asked to declare their valuation.

\begin{definition}[The Virtual Mechanism]\label{def:mech}
An influenced user $u_i$ declares some valuation to the company; then $u_i$ is sold the product at price $p_i$ if $p_i$ is no greater than the declared valuation, and not sold otherwise.
\end{definition}

\begin{theorem}[Truthfulness of the Virtual Mechanism]
The mechanism defined in Definition~\ref{def:mech} is truthful.
That is, the utility any user $u_i$ gets by declaring any number $\hat{v}_i \neq v_i$ is no greater than that she gets by declaring $v_i$ truthfully.
\end{theorem}

\begin{proof}
Consider the case that the true valuation $v_i < p_i$. Then, if reporting $v_i$ truthfully, $u_i$ will not adopt $p_i$, and hence her utility is $0$.
Suppose that $u_i$ reports lower ($\hat{v}_i < v_i$); she still would not get the product and the utility is still $0$.
Suppose otherwise ($u_i$ reports higher: $\hat{v}_i > v_i$).
Then, if $\hat{v}_i < p_i$, the situation is the same, in which $u_i$ gets zero utility.
If $\hat{v}_i \ge p_i$, then $u_i$ ends up paying $p_i$ to adopt and having a negative utility, since $v_i - p_i < 0$.

Then, consider the case that $v_i \ge p_i$, in which if $u_i$ reports truthfully, she will adopt the product by paying $p_i$ and enjoy a non-negative utility $v_i - p_i$.
Suppose that $u_i$ reports higher ($\hat{v}_i > v_i$), then she still gets to pay $p_i$ and has utility $v_i - p_i$.
Suppose otherwise ($u_i$ reports lower: $\hat{v}_i < v_i$).
Then, if $\hat{v}_i$ is still no less than $p_i$, she still pays $p_i$ and has utility $v_i - p_i$, while if $\hat{v}_i$ happens to be lower than $p_i$, she will not buy it and has zero utility.
And this completes the proof.
\end{proof}

\subsection{A Restricted Special Case of \promax{} under LT-V}\label{sec:basic}

To better understand the properties of the LT-V model and the
	hardness of \promax{}, we first study a special case
	of the problem.
We assume the valuation distributions degenerate to an identical single-point, i.e., $\forall u_i\in V$, $v_i = p$ for some $p \in (0,1]$ with probability $1$.
As mentioned in Sec.~\ref{sec:intro}, this is usually not the case; 
	the degeneration assumption here is of theoretical interest only.

For simplicity, we also assume that for every $u_i \in S$, the quoted price $p_i = 0$\footnote{Strictly speaking, for the sake of maximizing expected total profit, the seller should also charge price $p$ to all seeds, since it is assumed that all users have valuation $p$.  In this case, the expected profit function becomes $\hat{\pi}(S) = p\cdot h_L(S) - c_a |S|$, and the non-monotonicity still holds in general. To see this, consider a social network graph $G=(V,E)$ where $|V|=100$, and also let $p=0.5$, $c_a=0.1$. Suppose that there is a single node $v$ that has expected influence spread of $90$, which alone would yield a profit of $44.9$, while if $S=V$, the expected profit would be $40$, which is less than the case when $S=\{v\}$. This example illustrates that the profit function is still non-monotone when all users have the same valuation and are offered the same price.}.   
Since valuation is the maximum money one is willing
	to pay for the product, in this case, the optimal pricing
	strategy is to set $p_j = p$, $\forall u_j\in V\setminus S$. 
The situation amounts to restricting the marketing strategy to a binary one: free 
sample ($p_i = 0$) for seeds and full price for non-seeds ($p_j = p$).
Given this pricing strategy, once a node is {\em influenced}, it transitions to {\em adopting} with probability $1$.
Thus, \promax{} boils down to a problem to determine a seed
	set $S$, and the profit function $\pi(S,\p)$  reduces  to a set function $\pihat(S)$, since 
$\p$ is uniquely determined given $S$:
\begin{align}
\pihat(S) &= p \cdot (h_L(S) - |S|) -  c_a\, |S| \nonumber \\
	   &= p \cdot h_L(S) - (p + c_a)\, |S| ,
\end{align}
where $h_L(S)$ is the expected number of adopting nodes
	under the LT-V model by seeding $S$.

In general, the degenerated profit function $\pihat$ is {\em non-monotone}. 
To see this, let $u$ be any seed that provides a positive profit. Now, clearly 
$\pihat(\emptyset) = 0 < \pihat(\{u\})$ but $\pihat(V) \leq 0 < \pihat(\{u\})$, 
as giving free samples to the whole network will result in a loss of $c_a\,|V|$ on account of 
seeding expenses. 
Since $\pihat$ is non-monotone, unlike \infmax, it is natural to not use a budget $k$ for the number of seeds, 
but instead ask for a seed set of any size that results in the maximum expected profit. 
In other words, the number of seeds to be chosen,
	$k$, is not preset, but is rather determined by a solution.
This restricted case of \promax{} is to find $S=\argmax_{T\subset V}\pihat(T)$, which we show is \NPhard{}.

\begin{theorem}\label{thm:basicnp}
The Restricted \promax{} problem (RPM) is \NPhard{} for the LT-V model.
\end{theorem} 
 
\begin{proof}
Given an instance of the \NPhard{} {\sc Minimum Vertex Cover} (MVC) problem, we can construct an instance of the \promax{} problem, such that an optimal solution to the \promax{} problem gives an optimal solution to the MVC problem.
Consider an instance of MVC defined by an undirected $n$-node graph $G=(V,E)$; we want to find a set $S$ such that $|S|=k$ and $k$ is the smallest number such that $G$ has a vertex cover (VC) of size $k$.

The corresponding instance of RPM is as follows: first, we direct all edges in $G$ in both directions to obtain a directed graph $G'=(V,E')$, where $E'$ is the set of all directed edges.
Then, for each $u_i\in V$, set $\theta_i = 1$; for each $(u_i, u_j)\in E$, define $w_{i,j} = 1/\mathrm{d}^{\mathit{in}}(u_j)$, where $\mathrm{d}^{\mathit{in}}(u_j)$ is the in-degree of $u_j$ in $G'$.
Lastly, set $p=1$ and $c_a = 0$, in which case $\pihat(S) = h_L(S) - |S|$. 
Now, we want show that a set $S\subseteq V$ is a minimum vertex cover (MVC) of $G$ if and only if $S=\argmax_{T\subseteq V}\pihat(T)$.

$(\Longrightarrow)$.
If $S$ is a MVC of $G$, then in \promax{} we choose $S$ as the seed set, so that $\pihat(S) = n-|S|$.
This is optimal, shown by contradiction.
Suppose otherwise, i.e., there exists some $T\subseteq V$, $T\neq S$, such that $\pihat(T) > \pihat(S)$.
For the case of $|T| \ge |S|$, we have $\pihat(T) = h_L(T) - |T| \le h_L(T) - |S|$.
Since $h_L(T) \le n$, $\pihat(T) \le h_L(T) - |S| \le n-|S| = \pihat(S)$, which is a contradiction.
For the case of $|T| < |S|$, let $|S| - |T| = \ell$.
Thus, $\pihat(T) = h_L(T) - (|S|-\ell)$.
Since $T$ is not a VC, $h_L(S) = n$, and it is supposed that $\pihat(T) > \pihat(S)$, we have $h_L(T) = n-j$, for some $j \in [1,\ell)$.
Then, from the way in which influence weights and thresholds are set up, we know there are exactly $j$ nodes in $V\setminus T$ that are not activated.
Let $J$ be the set containing those $j$ nodes, and consider the set $T' = T\cup J$, for which we have $\pihat(T') = n$.
From the proof of Theorem 2.7 of \cite{kempe03}, $T'$ is a VC of $G$.
But since $|T'| = |T| + j < |S|$, $T'$ is a VC with a strictly smaller size than $S$, which gives a contradiction since $S$ is a MVC.

$(\Longleftarrow)$.
Suppose that $S=\argmax_{T\subseteq V}\pihat(T)$, but $S$ is not a VC of $G$ (we will consider MVC later).
This implies that there exists at least one edge $e\in E$ such that both endpoints of $e$, denoted by $u_i$ and $u_j$, are not in $S$.
From the way in which influence weights and thresholds are set up in $G'$, we know both $u_i$ and $u_j$ are not activated.
Thus, if we add either one of them, say $u_i$, into $S$, $h_L(S\cup \{u_i\})$ is at least $h_L(S) + 2$, and thus $\pihat(S\cup\{u_i\}) - \pihat(S) > 1$, which contradicts with the fact that $S$ optimizes $\pihat$.
Hence, $S$ must be a VC of $G$.
Now suppose that in addition $S$ is not a MVC.
Then, there must exist some $x\in S$ such that the node-set $S\setminus \{x\}$ is still a VC of $G$;
this means that $h_L(S\setminus\{x\}) = n$, too.
Thus, $\pihat(S\setminus \{x\}) = n - |S| + 1 > \pihat(S) = n-|S|$, which is a contradiction.
Hence, $S$ is indeed a MVC of $G$.

Now we have shown that an optimal solution to the restricted \promax{} problem is an optimal solution to the {\sc Minimum Vertex Cover} problem, and vice versa; this completes the proof.

\eat{We reduce the {\sc Minimum Vertex Cover} (MVC) problem to RPM.
Given an instance of MVC defined by an undirected graph $G = (V,E)$, we construct an instance of RPM by first directing all edges in both directions to obtain a directed graph $G'=(V,E')$.
Then, for each $u_i\in V$, set $\theta_i = 1$; for each $(u_i, u_j)\in E'$, define $w_{i,j} = 1/\mathrm{d}^{\mathit{in}}(u_j)$, where $\mathrm{d}^{\mathit{in}}(u_j)$ is the in-degree of $u_j$ in $G'$.
Lastly, set $p=1$ and $c_a = 0$. 
Then, we can show that a set $S$ of nodes in $G'$ is an optimal solution to RPM if and only if $S$ is a minimum vertex cover of $G$.
}

\end{proof}

Observe that both components of $\pihat$, $h_L(S)$ and $-|S|$, are submodular,
	which leads to the submodularity of $\pihat$ as it is a non-negative linear
	combination of two submodular functions. However, unlike for \infmax{}, 
	the function is non-monotone and we want to find a set $S$ of any size that 
	maximizes $\pihat(S)$, so the standard \greedy{} is not applicable here. 
In \cite{feige07}, Feige et al.\ gave a randomized local search ($2/5$-approximation) for maximizing general non-monotone submodular functions.
This is applicable to $\pihat$, but have time complexity $\mathrm{O}(|V|^3|E| / \epsilon)$, where $(1+\epsilon/|V|^2)$ is the per-step improvement factor in the search.
By contrast, the function $\pihat$ is the difference between a monotone submodular function and a linear function, we propose a greedy approach (Algorithm~\ref{alg:greedy2} \UG{}) with time complexity $\mathrm{O}(|V|^2 |E|)$ and a better approximation ratio, which is slightly lower than $1-1/e$.
\UG{} grows the seed set $S$ in a greedy fashion similar to \greedy{}, and terminates when no node can
	provide positive marginal gain w.r.t.\ $S$.

{
\begin{algorithm}[t!]
\caption{\UG{} ($G=(V,E)$, $\pihat$)}\label{alg:greedy2}
$S \gets \emptyset$\;
\While {true} {
	$u \gets \argmax_{u_i\in V\setminus S}\left[\pihat(S\cup\{u_i\}) - \pihat(S) \right]$\;
	\If{$\pihat(S\cup\{u\}) - \pihat(S) >0$} {
             $S \gets S \cup \{u\}$\; 
	}	
	\lElse{{\em break}\;}
}
Output $S$\;
\end{algorithm}
}

\begin{theorem}\label{thm:ug}
\eat{Given as input $G=(V,E)$, for \promax{} with objective function $\pi_L$, let $S_g\subseteq V$ be the seed set returned by Algorithm~\ref{alg:greedy2}, and $S^*\subseteq V$ be the optimal solution.
Then,} 
Given an instance of the restricted \promax{} problem under the LT-V model consisting of a graph $G = (V, E)$ with edge 
weights and objective function $\pihat$, let $S_g\subseteq V$ be the seed set returned by 
Algorithm~\ref{alg:greedy2}, and $S^*\subseteq V$ be the optimal solution.
Then,
\begin{align}
\pihat(S_g) \ge (1-1/e)\cdot\pihat(S^*) - \Theta(\max\{|S_g|,|S^*|\}) \label{eqn:basicBound}.
\end{align}
\end{theorem}

\begin{proof} 
Case $(i)$.\ If $|S^*| \le |S_g|$, then since $h_L$ is monotone and submodular, $h_L(S_g)\ge (1-1/e)\cdot h_L(S^*)$.
Thus, by the definition of $\pihat$, we have 
\begin{align*}
\pihat(S_g)  &=	p \cdot h_L(S_g) - (p+c_a)\, |S_g|  \\
		   &\ge	p (1-1/e)\cdot h_L(S^*) - (p+c_a)\, |S_g|  \\
		   &= 	(1-1/e)\cdot \pihat(S^*) - (p+c_a)\,|S_g| + (1-1/e)(p+c_a)\,|S^*|  \\
		   &= 	(1-1/e)\cdot \pihat(S^*) - \Theta(S_g).
\end{align*}

Case $(ii)$.\ If $|S^*| > |S_g|$, consider a set $S'_g$ obtained by running \UG{} until $|S'_g| = |S^*|$.
Clearly, from case $(i)$, we have $\pihat(S'_g)\ge (1-1/e)\cdot\pihat(S^*) - \Theta(|S'_g|)$.
Due to the fact that $|S^*| = |S'_g| > |S_g|$, and $S_g$ is obtained by running \UG{} until no node can provide positive marginal profit, we have $\pihat(S_g)\ge \pihat(S'_g) \ge (1-1/e)\cdot\pihat(S^*) - \Theta(|S^*|)$.
Combining the above two cases gives Eq.~\eqref{eqn:basicBound}.
\end{proof}

Theorem~\ref{thm:ug} indicates that
	the gap between the \UG{} solution and a $(1-1/e)$-approximation
	grows linearly w.r.t.\ the cardinality of the seed set. 
Since this cardinality is typically much smaller than the expected spread,
	\UG{} can provides quality guarantees for restricted \promax{}
	with objective function $\pihat$.

\subsection{Properties of the LT-V Model in the General Case}\label{sec:property}

Theorem~\ref{thm:basicnp} shows that in a restricted setting
	where exact valuations are known and the optimal pricing strategy
	is trivial, \promax{} is still \NPhard{}.
Now we consider the general \promax{} described in Sec.~\ref{sec:modeldef},
	and show that for any fixed price vector, the
	general profit function maintains submodularity (w.r.t.\ the seed
	set) regardless of the specific forms of the valuation distributions.

Given a seed set $S$ and a price vector $\p$, let $ap(u_i | S,\p)$ denote 
	$u_i$'s {\em adoption probability}, defined as the probability that $u_i$
	adopts the product by the end of the diffusion started with seed set $S$ and 
price vector $\p$.
Similarly, let $ip(u_i |S, \p_{-i})$ denote $u_i$'s
	{\em probability of getting influenced} under the same 
initial conditions, where
	$\p_{-i}\in[0,1]^{|V|-1}$ is the vector of all prices excluding $p_i$.
Also, let $\pi^{(i)}(S,\p)$ be the expected profit earned from $u_i$.
By model definition, for any $u_i\in V\setminus S$, we have
	$ap(u_i | S,\p) = ip(u_i |S, \p_{-i}) \cdot (1-F_i(p_i))$ and
	$\pi^{(i)}(S,\p) = p_i \cdot ap(u_i | S,\p)$.  
If $u_i\in S$, $ip(u_i |S, \p_{-i}) =1$ and $\pi^{(i)}(S,\p) = p_i \cdot (1-F_i(p_i)) - c_a$.

By linearity of expectations, we have $\pi(S,\p) = \sum_{u_i\in V} \pi^{(i)}(S,\p)$.
Hence, to analyze the profit function, we just need to focus on
	the adoption probability, in which the factor $(1-F_i(p_i))$ does not depend on $S$,
	but $ip(u_i |S, \p_{-i})$ calls for careful analysis, which we will present
	in the proof of Theorem~\ref{thm:submod}.

\eat{
\textcolor{blue}{
If $u_i\in S$, then $ip(u_i | S, \p_{-i}) = 1$, and thus $ap(u_i | S,\p) = 1-F_i(p_i)$.
Due to user's random transition process from {\em influenced} to
	{\em adopting}, the set of {\em adopting seeds} is a random
	subset of $S$, denoted by $R$, whose true realization can be
	any $T\subseteq S$ with $\Pr[R = T] = \prod_{u_i\in T} (1-F_i(p_i)) \cdot \prod_{u_j\in S\setminus T} F_j(p_j).$
In other words, $R$'s realization is determined by the outcomes
	of a series of Poisson trials, in which each seed performs
	a Bernoulli trial with success probability being its adoption
	probability, i.e.,  $1-F_i(p_i)$.
Fixing a certain $T\subseteq S$ as the true realization of $R$, we say that
	it corresponds to a {\em possible world} $X_T$. 
In a particular $X_T$, for any $u_i\in V\setminus S$,
	let $ip^{X_T}(u_i | T, \p_{-i})$ be the probability that $u_i$
	gets influenced when $T$ is the set of adopting seeds.
Averaging over all possible worlds, we have
	$ip(u_i | S, \p_{-i}) = \sum_{T\subseteq S}\Pr[R=T] \cdot ip^{X_T}(u_i | T, \p_{-i})$.
}}

Let $\vals = (v_1,\dots,v_{|V|}) \in [0,1]^{|V|}$ be a vector of 
	user valuations, corresponding to random samples drawn from the 
	various user valuation distributions. 
We now have:
\begin{theorem}[Submodularity] \label{thm:submod}
Given an instance of the LT-V model, for any fixed vector $\p\in[0,1]^{|V|}$ of prices, 
the profit function $\pi(S,\p)$ is submodular w.r.t.\ $S$, for an arbitrary 
vector $\vals$ of valuation samples.
\end{theorem}

The proof of submodularity of the influence spread function $h$
	in the classical LT model \cite{kempe03} relies on establishing 
an equivalence between the LT model and reachability
	in a family of random graphs generated as follows: for each node $u_i\in V$,
	select at most one of its incoming edges at random, such that
	$(u_j,u_i)$ is selected with probability $w_{j,i}$, and no edge is selected
	with probability $1-\sum_{u_j\in \Nin(u_i)} w_{j,i}$.
We will use a similar approach in the proof of Theorem~\ref{thm:submod}.

\begin{proof}[Proof of Theorem~\ref{thm:submod}]
By linearity of expectation as well as the above analysis on adoption probabilities, $\pi(S,\p) = \sum_{u_i \in V} \pi^{(i)}(S,\p) =  \sum_{u_i \in S} [p_i (1-F_i(p_i)) - c_a] + \sum_{u_i\not\in S} p_i (1-F_i(p_i))\cdot ip(u_i | S, \p_{-i})$.
Since the first sum is linear in $S$, it suffices to show that $ip(u_i | S, \p_{-i})$ is submodular in $S$, whenever $u_i \not \in S$.

To encode random events of the LT-V model using the possible world semantics, we do the following.
First, we run a {\em node coloring} process on $G$: for each node $u_i$, if $p_i\le v_i$, color it black; otherwise color it white.
Meanwhile, we run a {\em live-edge selection} process following the aforementioned protocol~\cite{kempe03}.
Note that the two processes are orthogonal and independent of each other.
Combining the results of both leads to a {\em colored live-edge} graph, which we call a {\em possible world} $X$.
Let $\X$ be the probability space in which each sample point specifies one such possible world $X$.

Next, we define the notion of ``black-reachability''.
In any possible world $X$, a node $u_i$ is {\em black-reachable} from a node set $S$ if and only if there exists a black node $s\in S$ such that $u_i$ is reachable from $s$ via a path consisting entirely of black nodes, except possibly for $u_i$ (even if $u_i$ is white, it is still considered black-reachable since here we are interested in the probability of being {\em influenced}, not {\em adopting}).
From the same argument in the proof of Claim 2.6 of \cite{kempe03}, on any black-white colored graph, the following two distributions over the sets of nodes are the same:
(1) the distribution over sets of {\em influenced} nodes obtained by running the LT-V process to completion starting from $S$;
(2) the distribution over sets of nodes that are {\em black-reachable} from $S$, under the live-edge selection protocol.

Let $I_X(u_i | S)$ be the indicator set function such that it is $1$ if $u_i$ is black-reachable from $S$, and $0$ otherwise.
Consider two sets $S$ and $T$ with $S\subseteq T\subseteq V$, and a node $x\in V\setminus T$.
Consider some $u_i$ that is black-reachable from $T\cup \{x\}$ but not from $T$.
This implies
(1) $u_i$ is not black-reachable from $S$ either (otherwise, $u_i$ would also be black-reachable from $T$, which is a contradiction);
(2) the source of the path that ``black-reaches'' $u_i$ must be $x$.
Hence, $u_i$ is black-reachable from $S\cup\{x\}$, but not from $S$, which implies $I_X(u_i | S\cup\{x\}) - I_X(u_i |S) = 1 \geq 1 = I_X(u_i | T\cup\{x\}) - I_X(u_i | T)$.
Thus, $I_X(u_i | S)$ is submodular.
Since $ip(u_i|S,\p_{-i}) = \sum_{X\in \X}\Pr[X]\cdot I_X(u_i | S)$ is a nonnegative linear combination of submodular functions, 
this completes the proof.
\end{proof}

We also remark that in general graphs, given any $S$ and $\p$, it is \SPhard\ to compute the exact
	value of $\pi(S,\p)$ for the LT-V model, just as in the case of computing the exact expected
	spread of influence for the LT model.
This can be shown using a proof similar to the one for Theorem 1 in~\cite{ChenYZ10}.

\section{Profit Maximization Algorithms}\label{sec:algo}
For \promax{}, since the expected profit is a function of both the seed set and the vector of prices, a \promax{} algorithm should determine both the seed set and an assignment of prices to nodes to optimize the expected profit.
Accordingly, it has two components: 
	(1).\ a seed selection procedure that determines $S$, and
	(2).\ a pricing strategy that determines $\p$.
Due to acquisition costs and the possible need for {\em seed-discounting}
	(details later), $\pi(S,\p)$ is still
	non-monotone in $S$ and is in the form of the difference between 
	a monotone submodular function and a linear function.
Hence, inspired by the restricted \promax{} studied in~\ref{sec:basic},
	we propose to use \UG{} for seed set selection.

We then propose three pricing strategies and integrate
	them with \UG{} to obtain three \promax{} algorithms.
The first two, \allomp{} and \ffs{}, are baselines with simple
	strategies that set prices of seeds without considering
	the network structure and influence spread, while the third one,
	\page{}, computes optimal discounts for candidate seeds based on their ``profit potential''. 
Intuitively, it ``rewards'' seeds with higher influence spread by giving them a deeper discount to boost their adoption probabilities, and in turn the adoption probabilities of nodes that may be influenced directly or indirectly by such seeds.

Notice that taking valuations into account when modeling the diffusion process of product adoption makes a difference for a marketing company.
A pricing strategy that does not consider valuations is limited: either it charges everyone full price (or at best gives full discount to the seeds), or it uses an ad-hoc discount policy which is necessarily suboptimal.
By contrast, PAGE makes full use of valuation information to determine the best discounts.
\subsection{Two Baseline Algorithms: \allomp{} and \ffs{}}

Recall that in our model, users in the social network are price-takers who
	myopically respond to the price offered to them.
Thus, given a distribution function $F_i$ of valuation $v_i$,
	the {\em optimal myopic price} (OMP)~\cite{Hartline08} can be
	calculated by:

\begin{align}\label{eqn:omp}
p^m_i = \argmax_{p\in[0,1]} \; p\cdot (1-F_i(p)) . 
\end{align}

Offering OMP to a {\em single} influenced node ensures that the
	expected profit earned {\sl solely from that node} is the maximum. 
This gives our first \promax{} algorithm, \allomp{},
	which offers OMP to all nodes regardless of whether a node
	is a seed or how influential it is.
First, for each $u_i\in V$, it calculates $p^m_i$ using
	Eq.~\eqref{eqn:omp}, and populates all OMPs to form
	the price vector $\p^m = (p_1^m, ..., p_{|V|}^m)$.
Then, treating $\p^m$ fixed, it essentially runs \UG{} (Algorithm~\ref{alg:greedy2})
	to select the seeds.
When the algorithm cannot find a node of which the marginal profit is positive, it stops. 

{
\begin{algorithm}[t!]
\caption{\allomp{} ($G=(V,E)$, $\pi$, $F_i (\forall u_i\in V)$)}\label{alg:omp}
$S \gets \emptyset$; $\p^m \gets \mathbf{0}$\;
\ForEach {$u_i \in V$} {
	$\p^m[i] \gets p^m_i = \argmax_{p\in[0,1]} \; p\cdot (1-F_i(p))$\;
}
\While {$true$} {
	$u \gets \argmax_{u_i\in V\setminus S} [\pi(S\cup\{u_i\},\p^m) - \pi(S,\p^m)]$\;
	\lIf{$\pi(S\cup\{u\},\p^m) - \pi(S,\p^m) > 0$}{
		$S\gets S\cup \{u\}$\;
	}\lElse{
		{\em break}\;
	}
}
Output $S,\p^m$\;
\end{algorithm}
}

Notice that Eq.~\eqref{eqn:omp} overlooks the network structure and ignores the profit potential of seeds.
This may lead to the sub-optimality of \allomp{} in general. 
	Fig~\ref{fig:quotient} illustrates this with an example.
\begin{figure}[h!t!p!]\label{fig:quotient}
	\centering
	\includegraphics[width=0.3\textwidth]{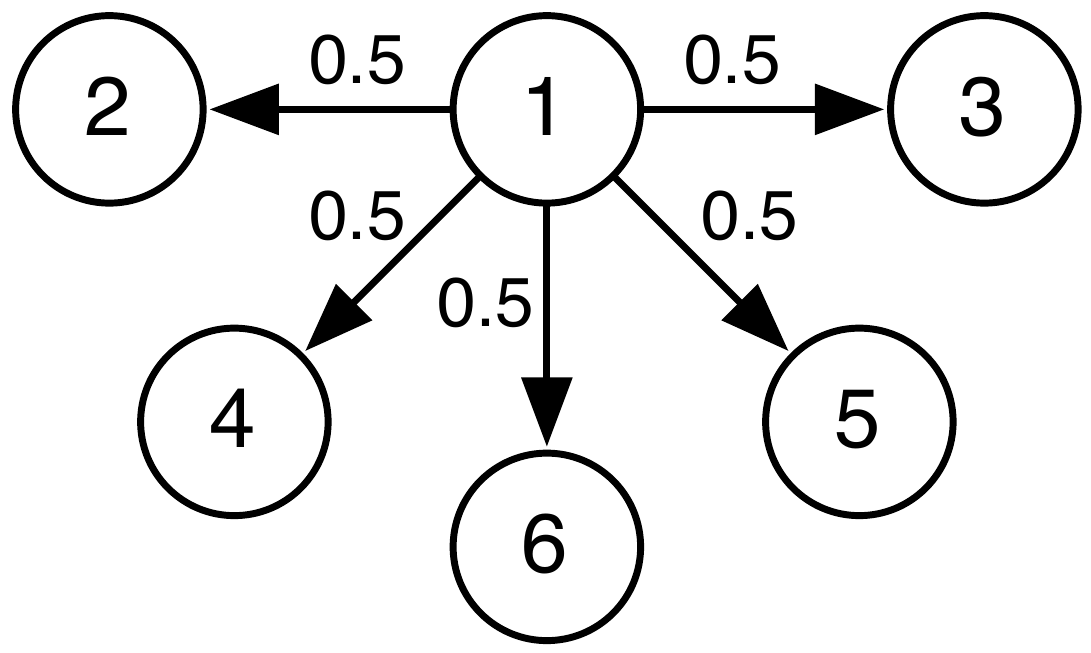}
	\caption{An example graph.}
	\label{fig:quotient}
\end{figure}
Suppose that all valuations are distributed uniformly in $[0,1]$ and the acquisition cost $c_a = 0.001$.
Hence, $\p^m = (1/2,\dotsc,1/2)$.
Consider seeding node $1$: it adopts w.p.\ $0.5$, giving a profit of $0.5+5*0.5^3-0.001 = 1.124$; it does not adopt w.p.\ $0.5$, resulting in a profit of $-0.001$.
Thus, the expected profit $\pi(\{1\},\p^m) = 0.5615$.
However, when $p_1 = 3/16$, $\pi(\{1\},\p^m_{-1}\oplus (3/16)) = 0.661$\footnote{We use $\p_{-i}\oplus x$ to denote a vector sharing all values with $\p$ except that the $i$-th coordinate is replaced by $x$, e.g., if $\p = (0.2, 0.3, 0.4)$, then $\p_{-1}\oplus 0.5 = (0.5, 0.3, 0.4)$.}.
This shows that for high-influence networks and low acquisition cost,
	the profit earned by running \allomp{} can be improved by
	{\em seed-discounting}, i.e., lowering prices for seeds so as
	to boost  their adoption probabilities and thus better leverage their
	influence over the network.
The intuition is that the profit loss over seeds (stemming from the discount) can potentially be
	compensated and even surpassed by the profit gain over non-seeds: more 
seeds may adopt as a result of the discount and the 
	probabilities of non-seeds getting influenced will go up as more seeds
	adopt.

Generally speaking, there exists a trade-off between the immediate
	(myopic) profit earned from seeds and the potentially more
	profit earned from non-seeds.
Favoring the latter, we propose our second algorithm \ffs{} (Free-For-Seeds) 
	which gives a full discount to seeds and charges non-seeds the OMP.
\ffs{} first calculates $\p^m = (p_1^m, ..., p_{|V|}^m)$ using Eq.~\eqref{eqn:omp}.
Then it runs \UG{}: in each iteration, it adds to $S$ the node which provides
	the largest marginal profit when a full discount (i.e., price $0$) is given.
For all seeds added, their prices remain $0$;
the algorithm ends when no node can provide positive marginal profit.

\begin{algorithm}[t!]
\caption{\ffs{} ($G=(V,E)$, $\pi$, $F_i (\forall u_i\in V)$)}\label{alg:ffs}
$S \gets \emptyset$; $\p^f \gets \mathbf{0}$\;
\ForEach {$u_i \in V$} {
	$\p^f[i] \gets p^m_i = \argmax_{p\in [0,1]} \; p\cdot (1-F_i(p))$\;
}
\While {$true$} {
	$u \gets \argmax_{u_i\in V\setminus S} [\pi(S\cup\{u_i\},\p^f_{-i}\oplus 0) - \pi(S,\p^f)]$\;
	\If{$\pi(S\cup\{u\},\p^f_{-u}\oplus 0) - \pi(S,\p^f) > 0$}{
		$S\gets S\cup \{u\}$; $\p^f \gets \p^f_{-u}\oplus 0$\;
	}\lElse{
		{\em break}\;
	}
}
Output $S,\p^f$\;
\end{algorithm}

Since \ffs{} has a completely opposite attitude towards
	seed-discounting compared to \allomp{},
	intuitively, it should be suitable for high-influence
	networks and low acquisition costs, but it may be overly aggressive
	for low-influence networks and high acquisition costs.
For example, in Fig~\ref{fig:quotient}, the \ffs{} profit by seeding node $1$ is $0.625$,
	better than the \allomp{} profit $0.5615$.
But if all influence weights are $0.01$ instead of $0.5$, and $c_a= 0.01$,
	\allomp{} gives a profit of $0.246$, while $\ffs$ gives only $0.0025$.

\subsection{The \page{} Algorithm}\label{sec:page}

Both \allomp{} and \ffs{} are easy for marketing companies to
	operate, but they are not balanced and are not robust against different
	input instances as illustrated above by examples.
To achieve more balance, we propose the \page{} (for Price-Aware GrEedy) algorithm (Algorithm~\ref{alg:pag}).
\page{} also employs \UG{} to select seeds. It initializes all seed prices to their OMP values (Step 3). 
In each round, it calculates the best price for each
	candidate seed such that its marginal profit (MP) w.r.t.\ the chosen
	$S$ and $\p$ is maximized (Step 7); then it picks the node with
	the largest maximum MP (Step 8). It stops when it cannot find a seed with a 
positive MP (Step 11). 
For all non-seed nodes, \page{} still charges OMP. We next explain the details of 
determining the best price for a candidate seed. 

Given a seed set $S$, consider an arbitrary candidate seed
	$u_i \in V\setminus S$, with its price $p_i$ to be determined.
The marginal profit ($MP$) that $u_i$ provides w.r.t.\ $S$ with
	$p_i$ is $MP(u_i) = \pi(S\cup\{u_i\},\p_{-i}\oplus p_i) - \pi(S,\p_{-i}\oplus p_i^m)$,
	where $\p_{-i}$ is fixed.
The key task in \page{} is to find $p_i$ such that $MP(u_i)$
	is maximized.
Since $\pi(S,\p_{-i}\oplus p_i^m)$ does not involve $u_i$ and $p_i$,
	 it suffices to find $p_i$ that maximizes $\pi(S\cup\{u_i\},\p_{-i}\oplus p_i)$.

Seeding $u_i$ at a certain price $p_i$ results in two possible worlds:
	world $X_1^{(i)}$ with $\Pr[X_1^{(i)}] = 1-F_i(p_i)$, in which $u_i$ adopts,
	and world $X_0^{(i)}$ with $\Pr[X_0^{(i)}] = F_i(p_i)$, in which $u_i$ does not adopt.
In world $X_1^{(i)}$, the profit earned from $u_i$ is $p_i-c_a$ and let the expected
	profit earned from other nodes be $Y_1$.
Similarly, in world $X_0^{(i)}$, the profit from $u_i$ is $-c_a$ and let the expected profit
	from other nodes be $Y_0$.
Notice that $Y_1$ depends on the influence of $u_i$ but $Y_0$ does not.
Putting it all together, the quantity of $\pi(S\cup\{u_i\},\p_{-i}\oplus p_i)$ can be expressed as
	a function of $p_i$ as follows:
\begin{align}\label{eqn:gp}
g_i(p_i) = (1-F_i(p_i))\cdot (p_i+Y_1) + F_i(p_i)\cdot Y_0  -c_a \,.
\end{align}
Similarly to the expected spread of influence in \infmax{}, the exact values of $Y_1$ and $Y_0$ cannot
	be computed in PTIME (due to \SPhard{ness}~\cite{ChenYZ10}), but
	sufficiently accurate estimations can be obtained by Monte Carlo (MC) simulations.

Finding $p_i^* = \argmax_{p_i\in[0,1]} g_i(p_i)$ now depends on the specific
	form of the distribution function $F_i$.
We consider two kinds of distributions: the {\em normal}
	distribution, for which $v_i\sim \N(\mu,\sigma^2),\, \forall u_i \in V$, and
	the {\em uniform} distribution, for which $v_i\sim \U(0,1),\, \forall u_i\in V$.
The choice of the normal distribution is supported by evidence from real-world data	
	from \texttt{Epinions.com} (see Sec.~\ref{sec:exp}), and also work in~\cite{jiangklb05}.
When sales data are not available, it is common to consider the uniform
	distribution with support $[0,1]$ to account for our complete lack
	of knowledge~\cite{klbBook,bloch08}.

\paragraph{The Normal Distribution Case.}
For normal distribution, assume that $v_i\sim \N(\mu,\sigma^2)$ for some $\mu$ and $\sigma$, then $\forall p_i\in [0,1]$,
\[
F_i(p_i) = \frac{1}{2} \left[1 + \erf\left(\frac{p_i-\mu}{\sqrt{2}\sigma}\right)\right],
\]
	where $\erf(\cdot)$ is the error function, defined as
\[
\erf(x) = \frac{2}{\sqrt{\pi}} \int_0^x e^{-t^2}\,\mathrm{d}t.
\]
Plugging $F_i(\cdot)$ back into Eq.~\eqref{eqn:gp}, one cannot obtain an analytical
	solution for $p_i^*$, as $\erf(x)$ has no closed-form expression.
Thus, we turn to numerical methods to approximately find $p_i^*$.
Specifically, we use the {\em golden section search algorithm},
	a technique that finds the extremum of a unimodal
	function by iteratively shrinking the interval inside which
	the extremum is known to exist~\cite{nabook94}.
In our case, the search algorithm starts with the interval $[0,1]$, and we set the {\em stopping criteria}
	to be that the size of the interval which contains $p_i$ is strictly smaller than $10^{-8}$.

\paragraph{The Uniform Distribution Case.}
The uniform distribution has easier calculations and analytical solutions.
If $v_i\sim \U(0,1)$, then $\forall p_i\in [0,1]$, $F_i(p_i) = p_i$, and
	plugging it back to Eq.~\eqref{eqn:gp} gives
\[
g_i(p_i) = -p_i^2 + (1-Y_1+Y_0)\cdot p_i + Y_1 - c_a.
\]
Hence, the optimal price
\[
p_i^* = \frac{(1+Y_1-Y_0)}{2}.
\]

For both normal and uniform distributions, if $p_i^* > 1$ or $p_i^* < 0$,
	it is normalized back to $1$ or $0$, respectively.  
Also note that the above solution framework applies to {\sl any probability distribution} that $v_i$ may follow, as long as an analytical or numerical solution can be found for $p_i^*$.

To conclude this section, 
	steps~\ref{line:a}-\ref{line:b} in Algorithm~\ref{alg:pag} (and also the \UG{} seed selection procedure in \allomp{} and \ffs{})
	can be accelerated by the CELF optimization~\cite{Leskovec07},
	or the more recent CELF++ \cite{www_poster}.
The adaptation is straightforward and the details can be found in~\cite{Leskovec07} and \cite{www_poster}.

{
\begin{algorithm}[t!]
\caption{\page{} ($G=(V,E)$, $\pi$, $F_i (\forall u_i\in V)$)}\label{alg:pag}
$S \gets \emptyset$; $\p \gets \mathbf{0}$\;
\ForEach {$u_i \in V$} {
	$\p[i] \gets p^m_i = \argmax_{p\in [0,1]} \; p\cdot (1-F_i(p))$\;
}
\While {$true$} {
       \ForEach{$u_i \in V\setminus S$} {  \label{line:a}
		Estimate the value of $Y_0$ and $Y_1$ by MC simulations\;
		$p_i^* \gets \argmax_{p_i\in[0,1]} g_i(p_i)$; normalize if needed\;
	}
	$u \gets \argmax_{u_i\in V\setminus S} g_i(p_i^*)$\; \label{line:b}
	\If{$\pi(S\cup\{u_i\},\p_{-i}\oplus p_i^*) - \pi(S,\p_{-i}\oplus p_i^m) > 0$}{
		$S\gets S\cup \{u_i\}$; $\p \gets \p_{-i}\oplus p_i^*$\;
	}\lElse{
		{\em break}\;
	}
}
Output $S,\p$\;
\end{algorithm}
}

\section{Empirical Evaluations}\label{sec:exp}
\begin{table}
\centering
\caption{Statistics of Network Data.}
	\begin{tabular}{|p{4cm} || c | c | c |}
		\hline
		\textbf{Dataset} & \textbf{Epinions} & \textbf{Flixster} & \textbf{NetHEPT}   \\ \hline \hline
		Number of nodes & $11$K & $7.6$K & $15$K   \\ \hline
		Number of edges & $119$K & $50$K & $62$K    \\ \hline
	    Average out-degree & $10.7$ & $6.5$ & $4.12$   \\ \hline
	    Maximum out-degree & $1208$ & $197$ & $64$  \\ \hline
	    \#Connected components & $4603$  & $761$ & $1781$   \\ \hline
		Largest component size & $5933$ & $2861$ & $6794$   \\ \hline
	\end{tabular}
	\label{table:dataset}
\end{table}

We conduct experiments on real-world network datasets to evaluate 
	our proposed baselines and the \page{} algorithm.
In all these algorithms, a key step is to compute the marginal profit of a candidate seed.
As mentioned in Sec.~\ref{sec:model}, computing the exact expected profit is intractable for the LT-V model.
Thus, we estimate the expected profit with Monte Carlo (MC) simulations.
Following \cite{kempe03}, we run 10,000 simulations for this purpose.
This is an expensive step and as for \infmax{}, it limits the size of networks on which we can run these simulations. 
For the same reason, the CELF optimization is used in all algorithms as a heuristic.
All implementations are in C++ and all experiments were run on a
	server with $2.50$GHz eight-core Intel Xeon E5420 CPU,
	$16$GB RAM, and Windows Server 2008 R2.

\subsection{Dataset Preparations}

\paragraph{Network Data.}
We use three network datasets whose statistics are summarized in
	Table~\ref{table:dataset}.
They include:
(a) Epinions~\cite{richardson02}, a who-trust-whom network extracted from
	review site \texttt{Epinions.com}: an edge $(u_i, u_j)$ is present
	if $u_j$ has expressed her trust in $u_i$'s reviews;
(b) Flixster\footnote{\url{http://www2.cs.sfu.ca/~sja25/personal/datasets/}. Ratings timestamped.}, a 
	friendship network from social movie site \texttt{Flixster.com}:
	if $u_i$ and $u_j$ are friends, we have edges in both directions;
(c) NetHEPT (standard for \infmax{}~\cite{kempe03, ChenWW10, ChenYZ10, simpath})\footnote{\url{http://research.microsoft.com/en-us/people/weic/projects.aspx}},
	a co-authorship network extracted from the High Energy Physics Theory
	section of \texttt{arXiv.org}: if $u_i$ and $u_j$ have co-authored papers,
	we have edges in both directions.
The raw data of Epinions and Flixster contain $76$K users, $509$K edges and $1$M users, $28$M edges, respectively.
We use the METIS graph partition software\footnote{http://glaros.dtc.umn.edu/gkhome/views/metis} to extract a
	subgraph for both networks, to ensure that MC simulations can finish in a reasonable amount of time.

\begin{figure}[h!t!]
\centering
        \subfigure[Weighted Distribution]{ 
                \includegraphics[width=0.45\textwidth]{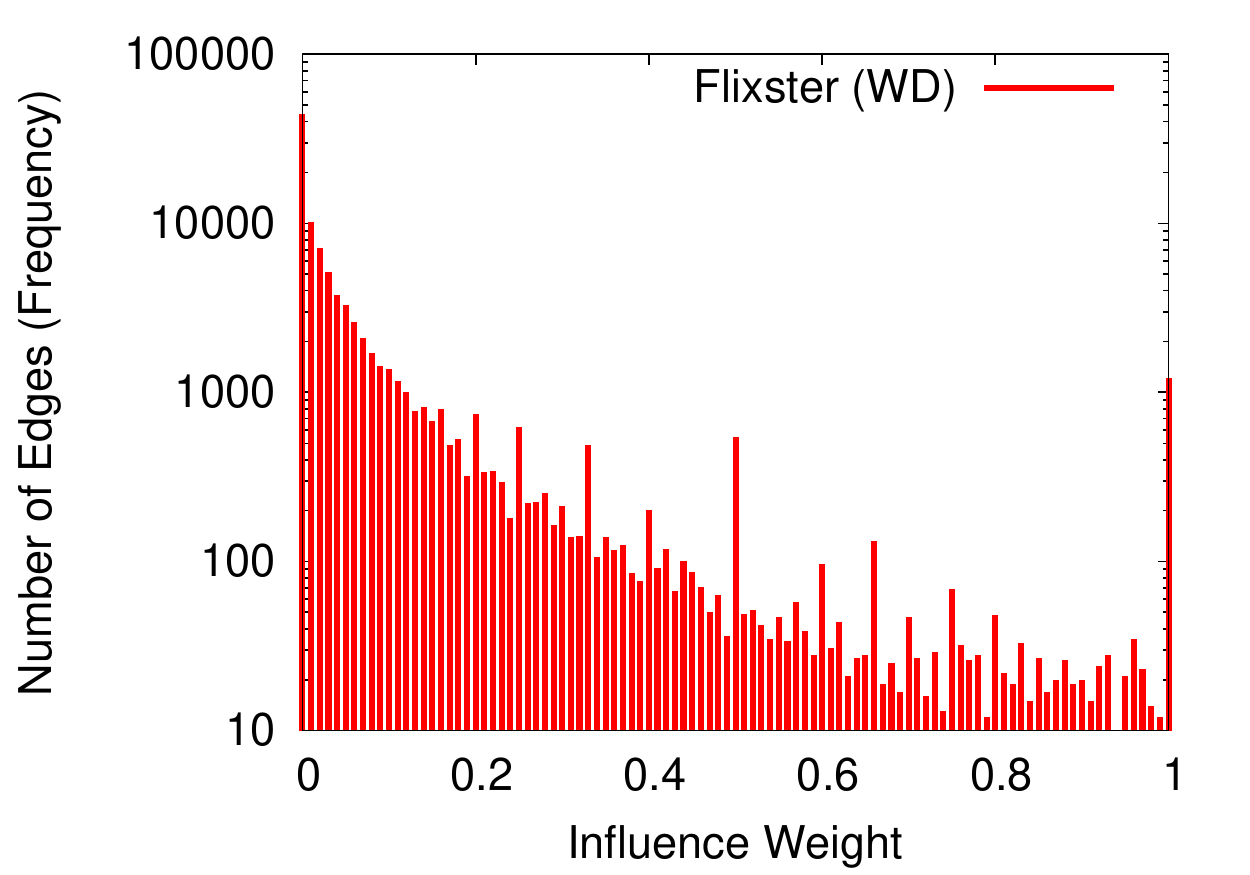}
                \label{fig:LT2dist}
	}
        \subfigure[Trivalency]{
                \includegraphics[width=0.45\textwidth]{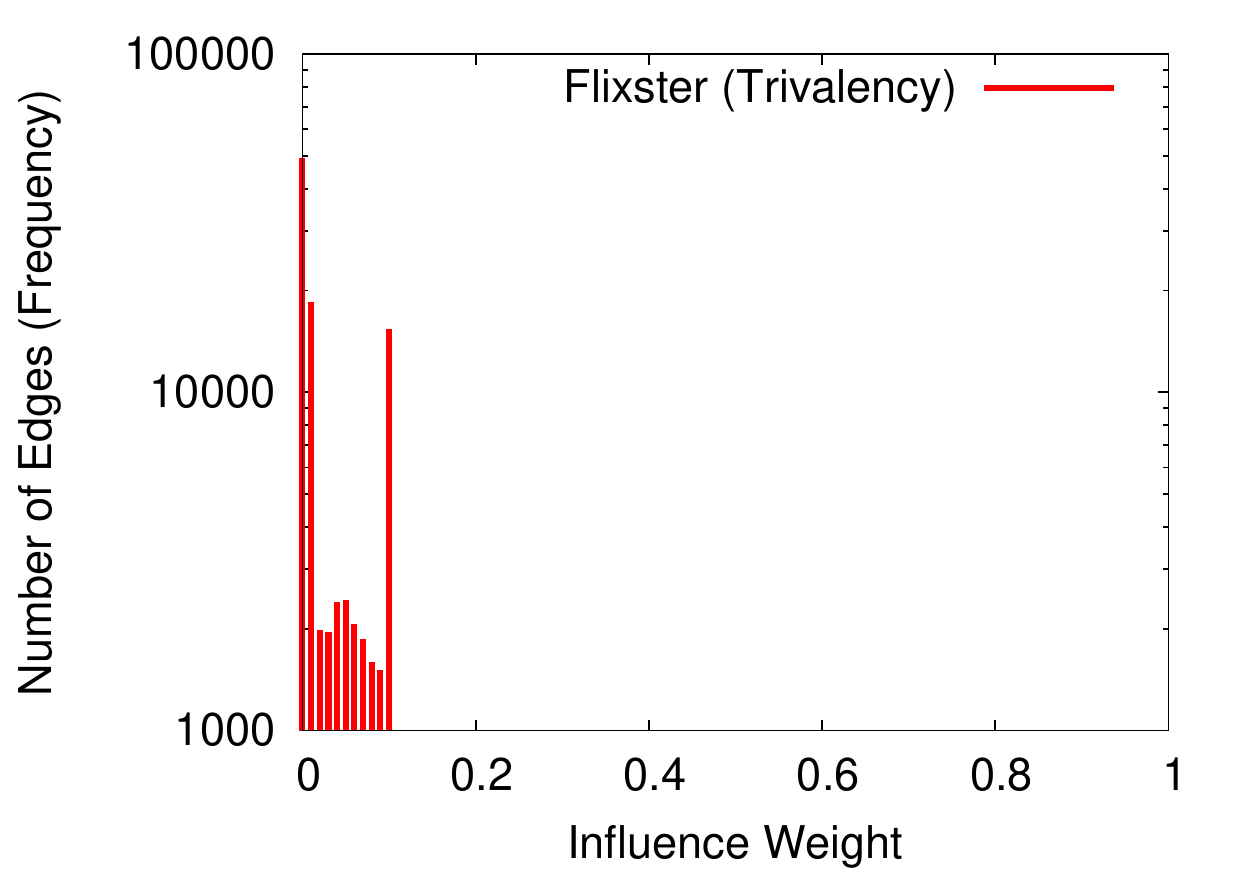}
                \label{fig:TVdist}
       }
        \caption{Distribution of influence weights in Flixster}
	\label{fig:flixdist}
\end{figure}

\paragraph{Influence Weights.}
We use two methods, Weighted Distribution (WD) and
	Trivalency (TV), to assign influence weights to edges.
For WD, $w_{i,j} = A_{i,j} / N_j$, where $A_{i,j}$ is the number of actions
	$u_i$ and $u_j$ both perform, and $N_j$ is a normalization factor, i.e., 
	the number of actions performed by $u_j$, 
	to ensure $\sum_{u_i\in \Nin(u_j)} w_{i,j} \le 1$.
In Flixster, $A_{i,j}$ is the number of movies $u_j$ rated after $u_i$;
	in NetHEPT, $A_{i,j}$ is the number of papers $u_i$ and $u_j$ co-authored;
	in Epinions, since no action data is available, we use $w_{i,j} = 1 / \mathrm{d}^{\mathit{in}}(u_j)$ as an approximation.
For TV, $w_{i,j}$ is selected uniformly at random from $\{0.001, 0.01, 0.1\}$, and is normalized to ensure $\sum_{u_i\in \Nin(u_j)} w_{i,j} \le 1$.
Fig.~\ref{fig:flixdist} illustrates the distribution of weights for Flixster; it shows that influence is higher in WD graphs than in TV graphs.

\paragraph{Valuation Distributions.}
As mentioned in Sec.~\ref{sec:intro}, valuations are difficult to
	obtain directly from users, and we have to estimate the distribution using
	historical sales data.
In an \texttt{Epinions.com} review, a user provides an integer rating from $1$ to $5$, and may optionally
	report the price she paid in US dollars (see, e.g., \url{http://tinyurl.com/773to53}).
If a review contains both price and rating, we can combine them to approximately estimate
	the valuation of that user, as in such systems, ratings are seen as people's utility for
	a good, and utility is the difference of valuation and price~\cite{klbBook}.

We observed that most products have only a limited number ($< 100$) of reviews, and thus a single product
	may not provide enough samples.
To circumvent this difficulty, we acquired all reviews for the popular Canon EOS 300D, 350D, and 400D
	cameras.
Given that these cameras followed a sequential release within a short time span (three years),
	we treated them as having similar monetary values to consumers.
After removing reviews without prices reported, we are left with $276$ samples.
Next, we transform prices and ratings to obtain estimated valuations as follows:
\[
\mathrm{valuation} = \mathrm{price} * (1 + \mathrm{rating} \, / \, 5).
\]

We then normalize the results into $[0,1]$ and fit the data to a normal
	distribution $\N(\mu,\sigma^2)$ with $\mu  = 0.53$ and $\sigma = 0.14$ estimated by maximum likelihood estimation (MLE).
Fig.~\ref{fig:eosPDF} plots the histogram of the normalized valuations;
Fig.~\ref{fig:eosKS} presents the CDFs of our empirical data and $\N(0.53, 0.14^2)$.
To test the goodness of fit, we compute the Kolmogorov-Smirnov (K-S) statistic~\cite{GonzalezSF77} of the two distributions, which is defined as the maximum difference between the two CDFs;
in our case, the K-S statistic is $0.1064$. 
As can be seen from Fig.~\ref{fig:eosKS}, $\N(0.53, 0.14^2)$ is indeed a good fit for the estimated
	valuations of the three Canon EOS cameras on \texttt{Epinions.com}.

Since there are no price data to be collected in Flixster and NetHEPT,
	we use $\N(0.53, 0.14^2)$ in the simulations for all datasets.
In addition, for completeness, we also test with the uniform distribution
	over $[0,1]$, i.e., $\U(0,1)$, as it is commonly assumed in the
	literature~\cite{klbBook,bloch08}.

\begin{figure}[h!t!]
\centering
       \subfigure[Histogram of valuations]{
                \includegraphics[width=0.45\textwidth]{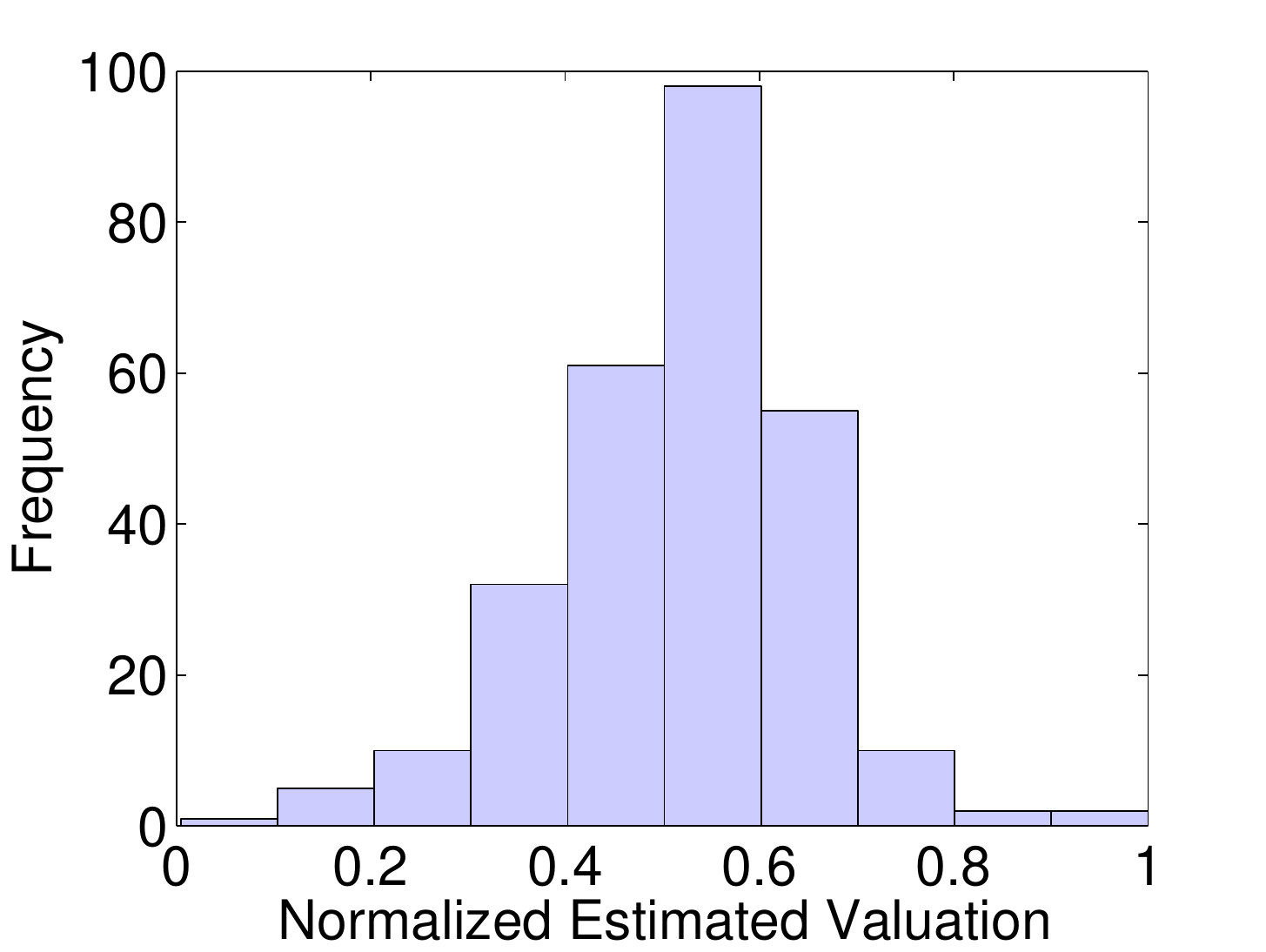}
                \label{fig:eosPDF}
        }
        \subfigure[Empirical \& normal CDFs]{
                \includegraphics[width=0.45\textwidth]{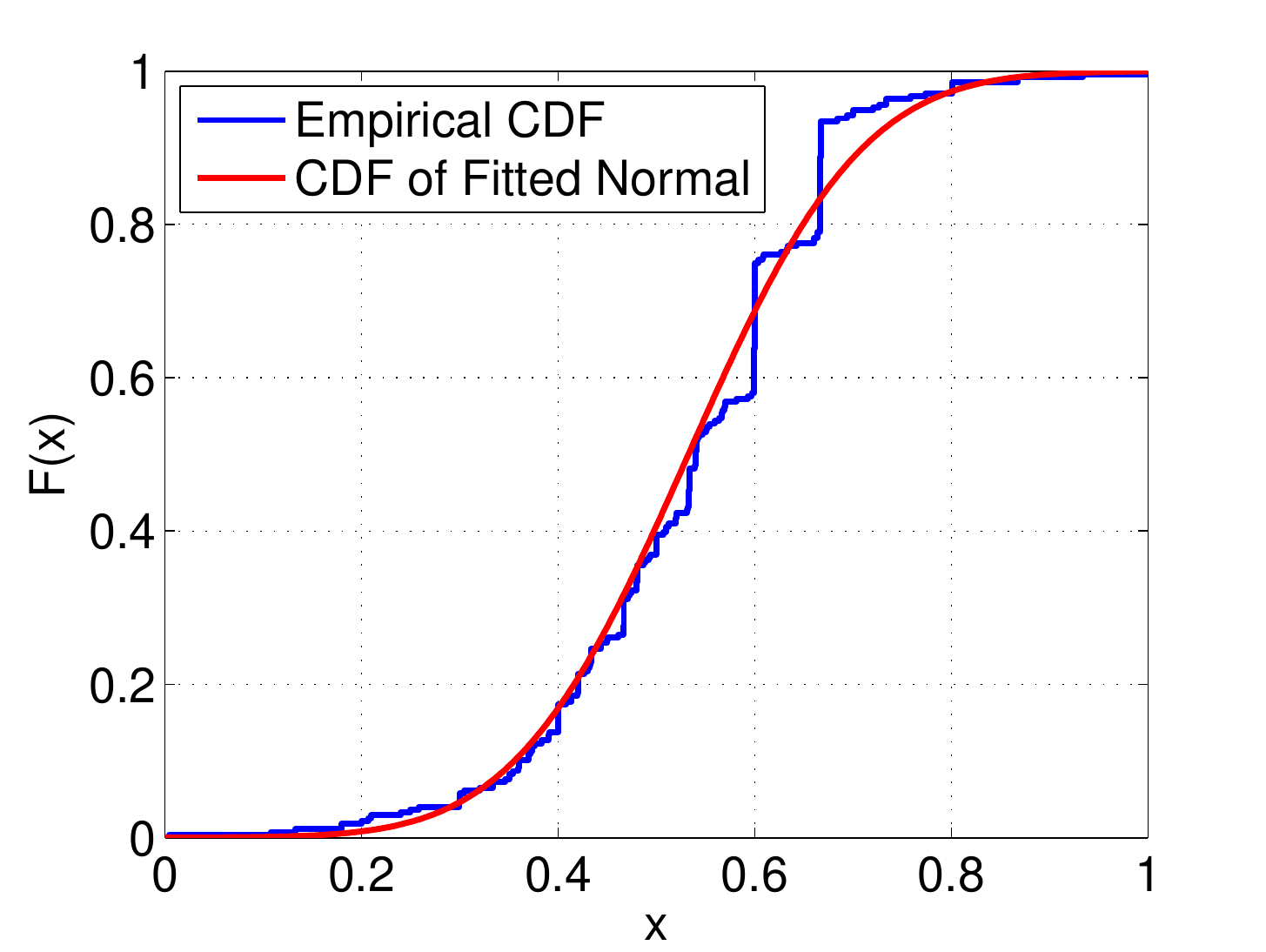}
                \label{fig:eosKS}
        }
        \caption{Statistics of Valuations (\texttt{Epinions.com})}
	\label{fig:eosdata}
\end{figure}

\subsection{Experimental Results}
We compare \page{}, \allomp{}, and \ffs{} in terms of the
	expected profit achieved, price assignments, and running time.
Although all algorithms employ \UG{} which does not terminate
	until the marginal profit starts decreasing, for
	uniformity, we report simulation results up to $100$ seeds.

\eat{
\begin{figure}[h!t!]
\centering
\includegraphics[width=0.5\textwidth]{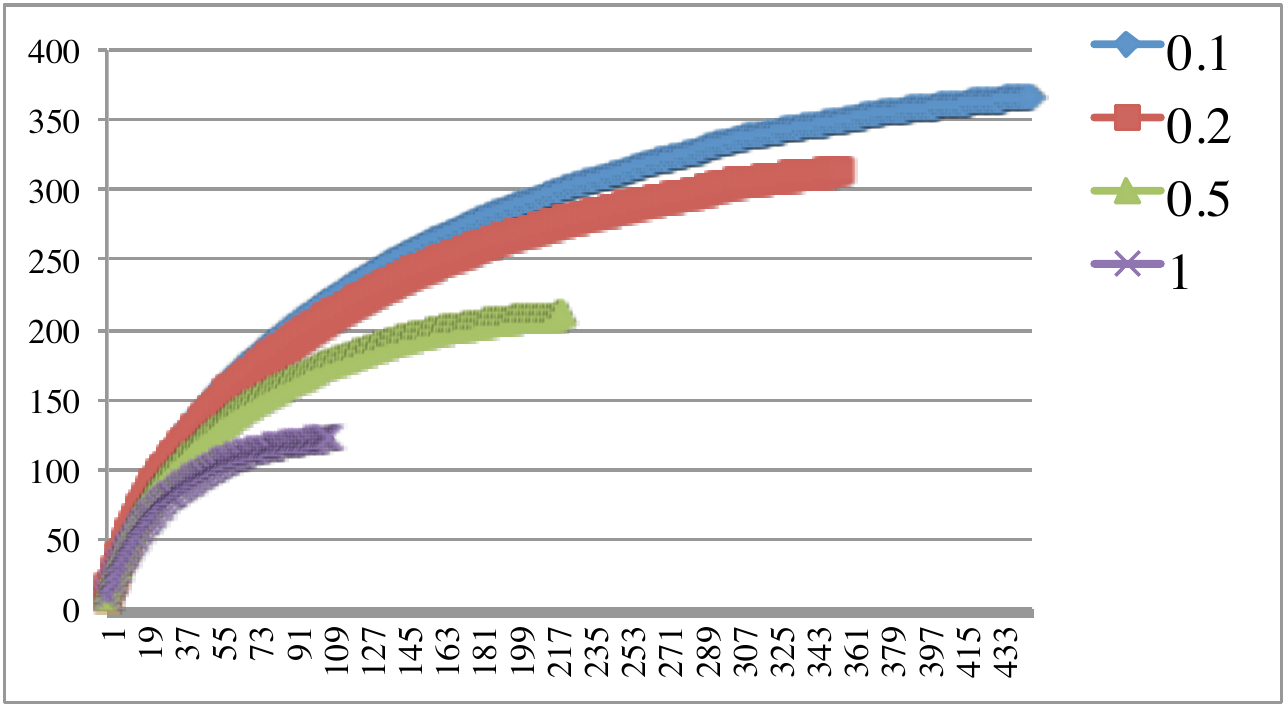}
\caption{\UG{} with various choices for $c_a$. X-axis: $|S|$; Y-axis: Expected profit achieved.}
\label{fig:hepBasic}
\end{figure}

\paragraph{\UG{}.}
To test the performance of \UG{} alone,
	we experiment with the function $\pihat$ for restricted \promax{}
	with $p=1$ and four acquisition costs: $c_a = 0.1, 0.2, 0.5$, and $1.0$
	(in reality, $c_a$ may not be as high; here we use these numbers just for exploration).
Fig.~\ref{fig:hepBasic} illustrates the expected values of $\hat{\pi}$ with all four costs
	on NetHEPT-TV.
In each case, the curve is cut off at the point where \UG{} stopped as $\hat{\pi}(S)$ started to decrease.
The larger $c_a$ is, the earlier \UG{} terminated, which is expected.
}

\begin{figure*}[h!t!]
\begin{tabular}{cc}
    \includegraphics[width=.45\textwidth]{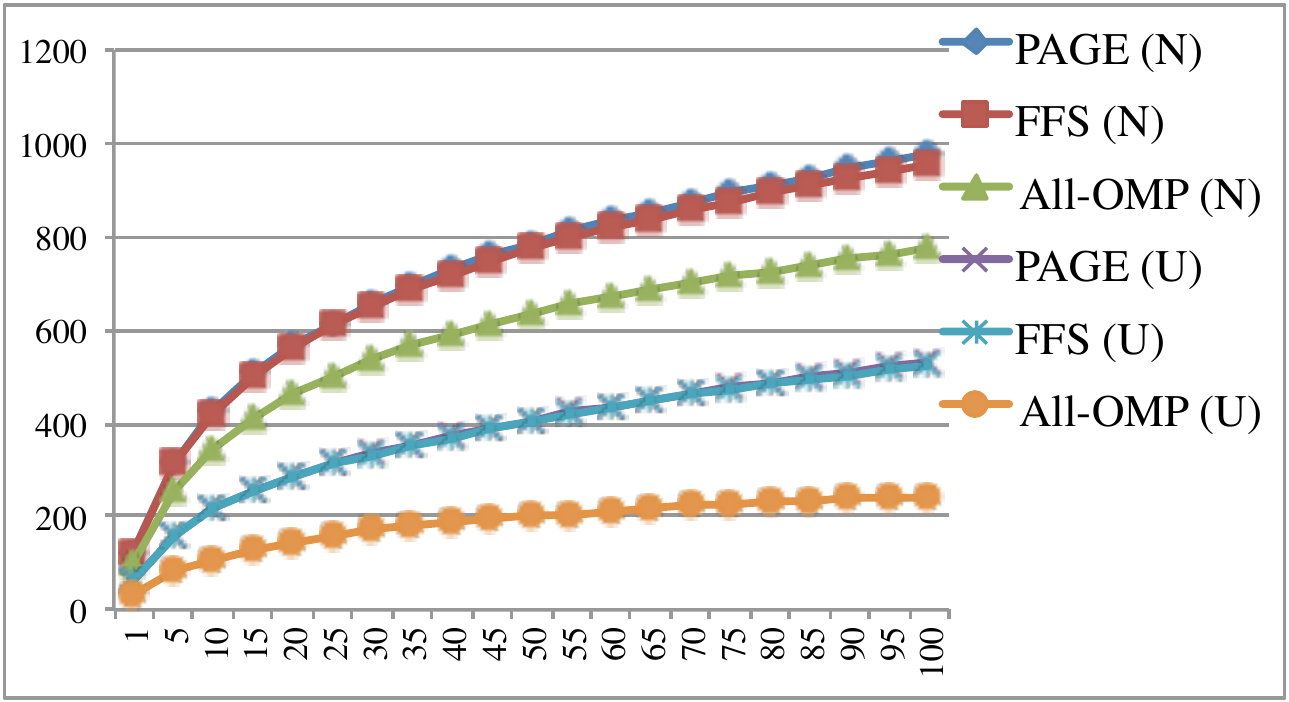}&
    \includegraphics[width=.45\textwidth]{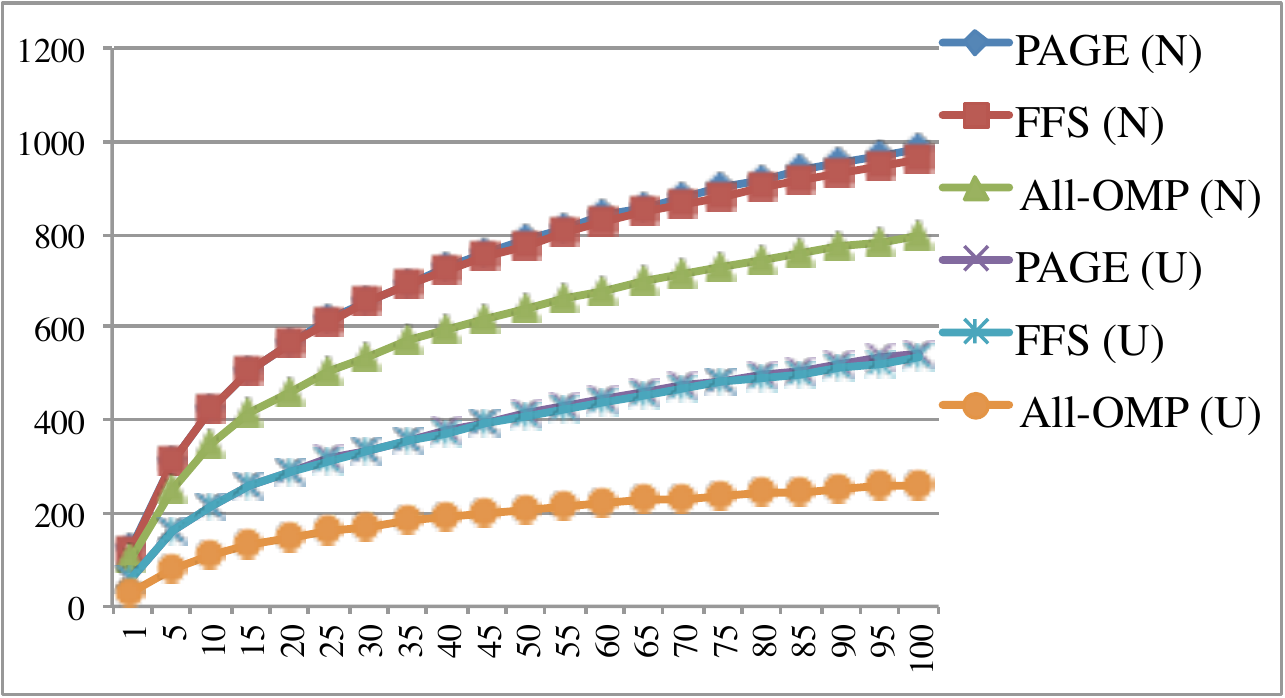} \\
    (a) {WD with $c_a = 0.1$}  & (b) {WD with $c_a = 0.001$} \\
    \includegraphics[width=.45\textwidth]{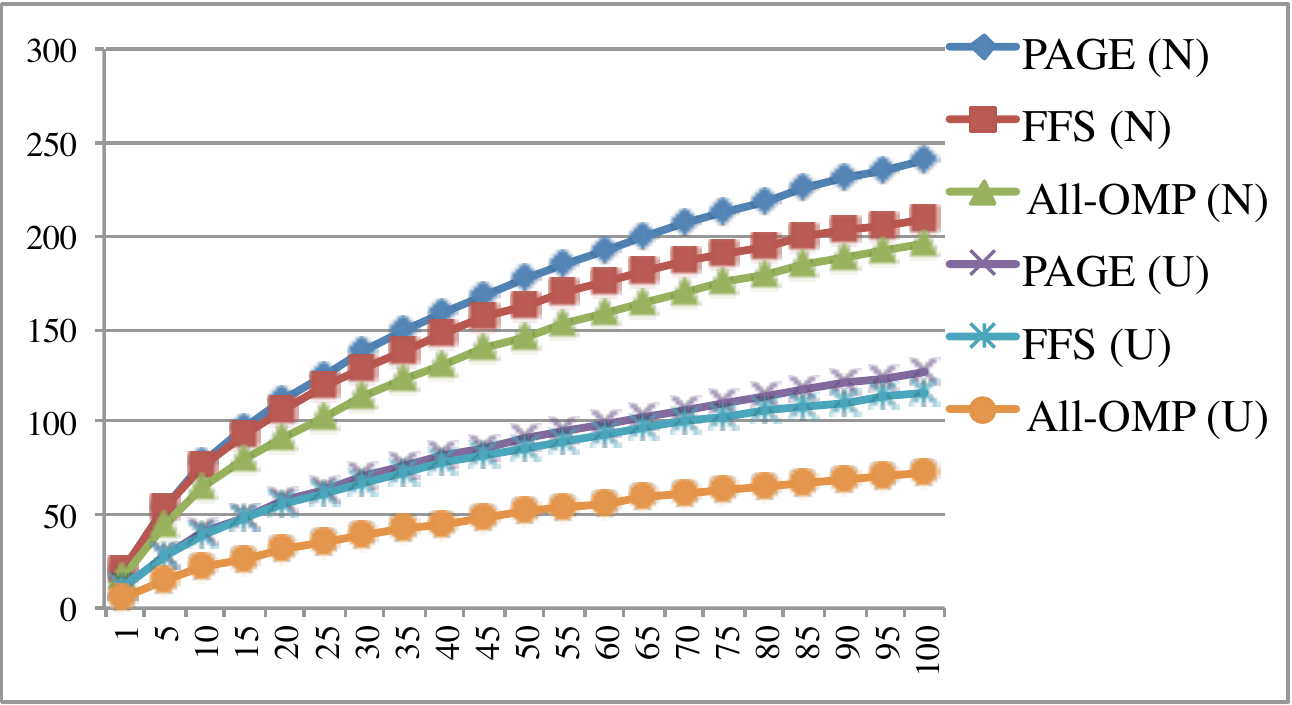}&
    \includegraphics[width=.45\textwidth]{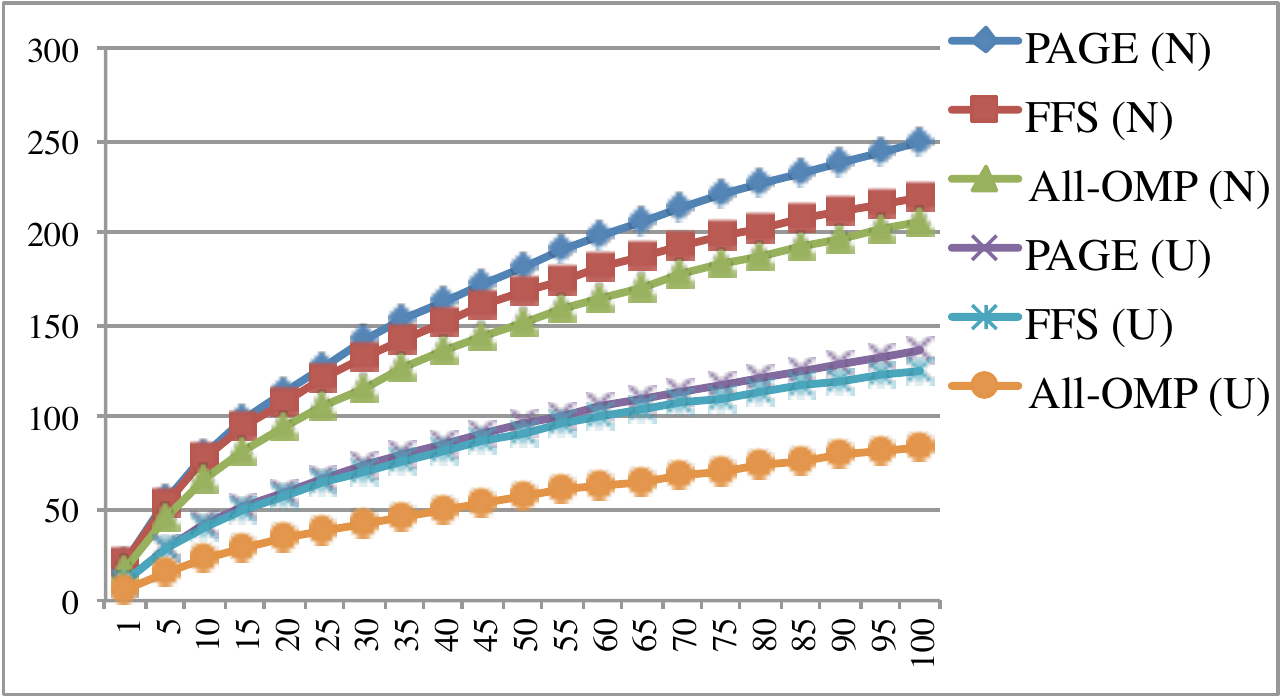}\\
    (c) {TV with $c_a = 0.1$} & (d) {TV with $c_a = 0.001$}  \\
\end{tabular}
\caption{Expected profit achieved (Y-axis) on Epinions graphs w.r.t.\ $|S|$ (X-axis). (N)/(U) denotes normal/uniform distribution.}
\label{fig:epiSpread}
\end{figure*}

\begin{figure*}[h!t!]
\begin{tabular}{cc}
    \includegraphics[width=.45\textwidth]{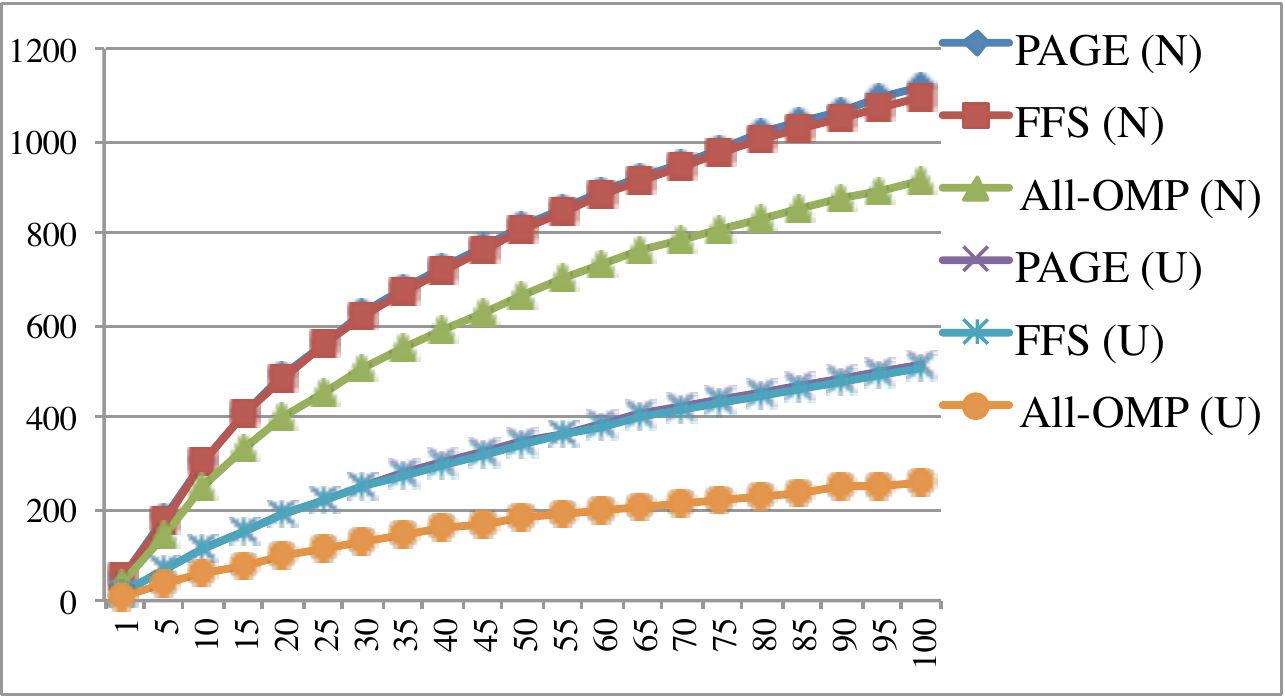} &
    \includegraphics[width=.45\textwidth]{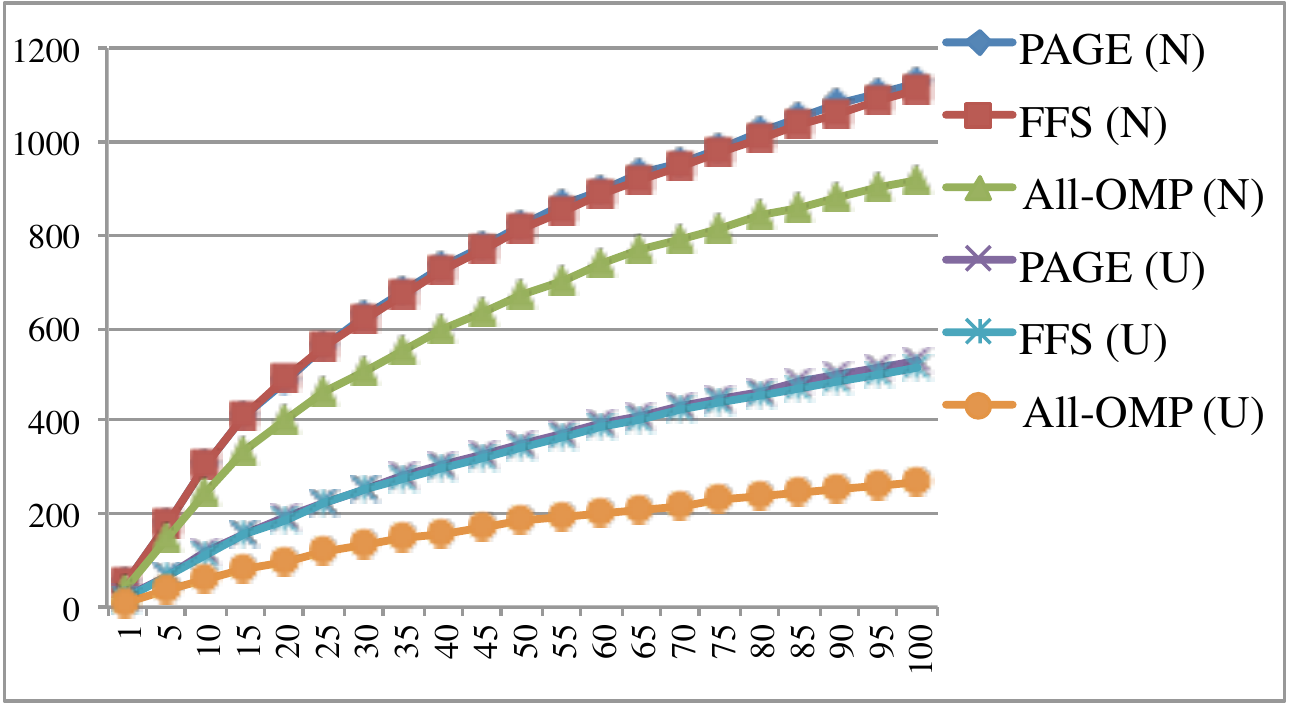} \\
	(a) {WD with $c_a = 0.1$}  & (b) {WD with $c_a = 0.001$} \\
    \includegraphics[width=.45\textwidth]{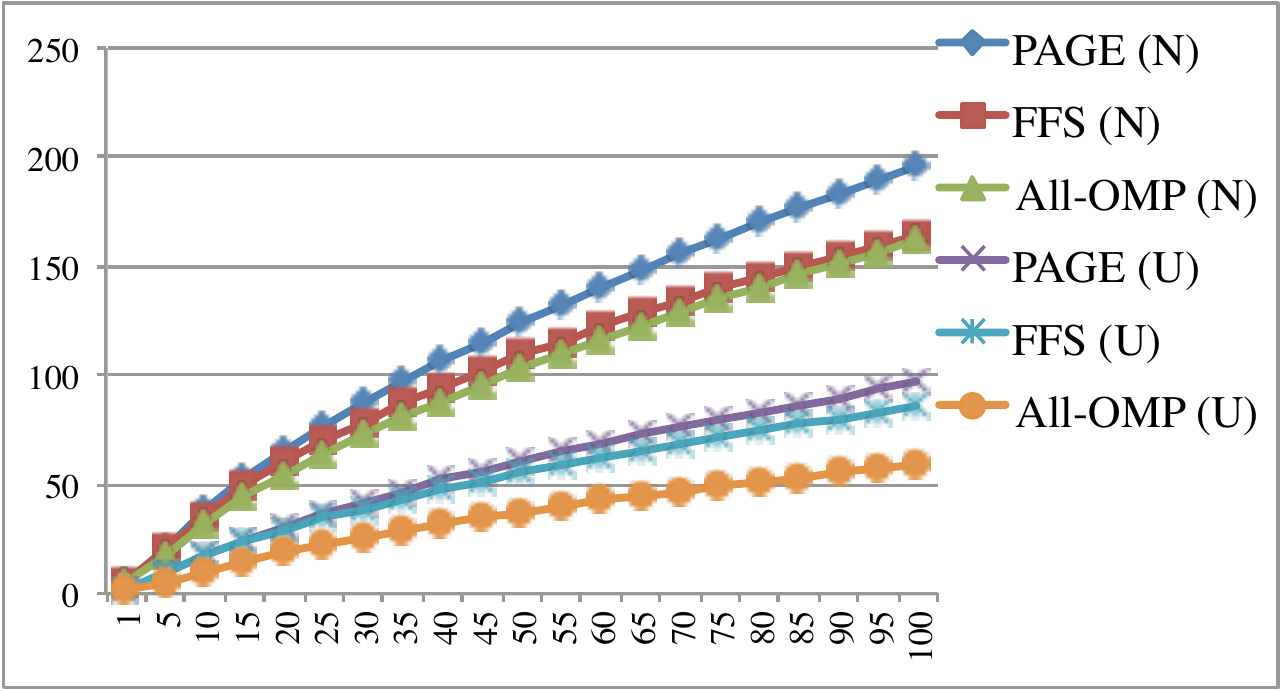}&
    \includegraphics[width=.45\textwidth]{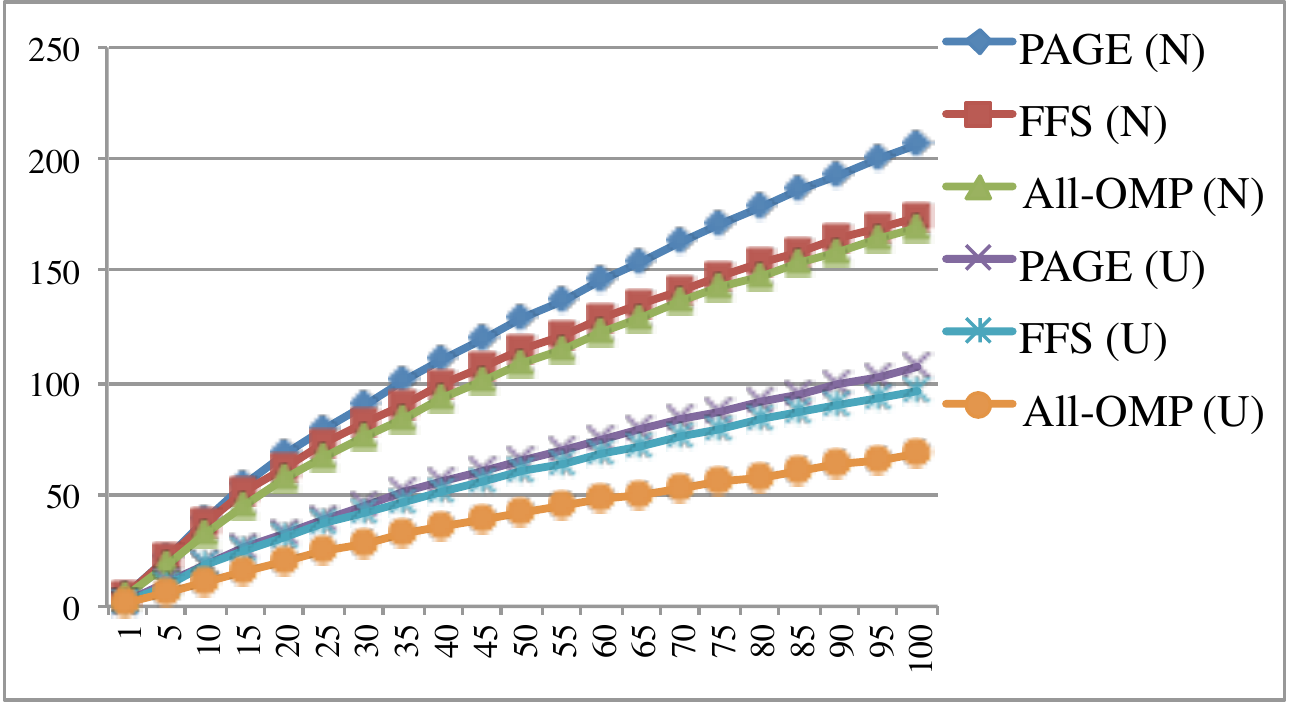}\\
	 (c) {TV with $c_a = 0.1$} & (d) {TV with $c_a = 0.001$}  \\
\end{tabular}
\caption{Expected profit achieved (Y-axis) on Flixster graphs w.r.t.\ $|S|$ (X-axis). (N)/(U) denotes normal/uniform distribution.}
\label{fig:flixSpread}
\end{figure*}

\begin{figure*}[h!t!]
\begin{tabular}{cc}
    \includegraphics[width=.45\textwidth]{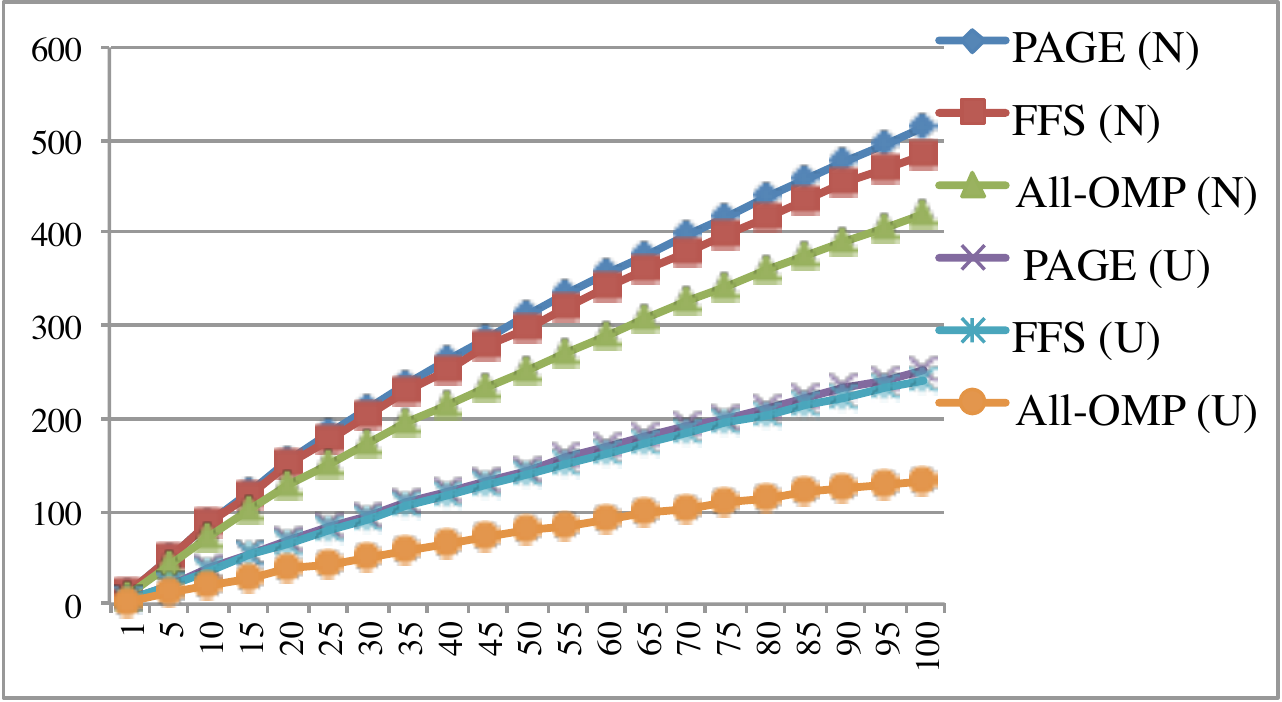}&
    \includegraphics[width=.45\textwidth]{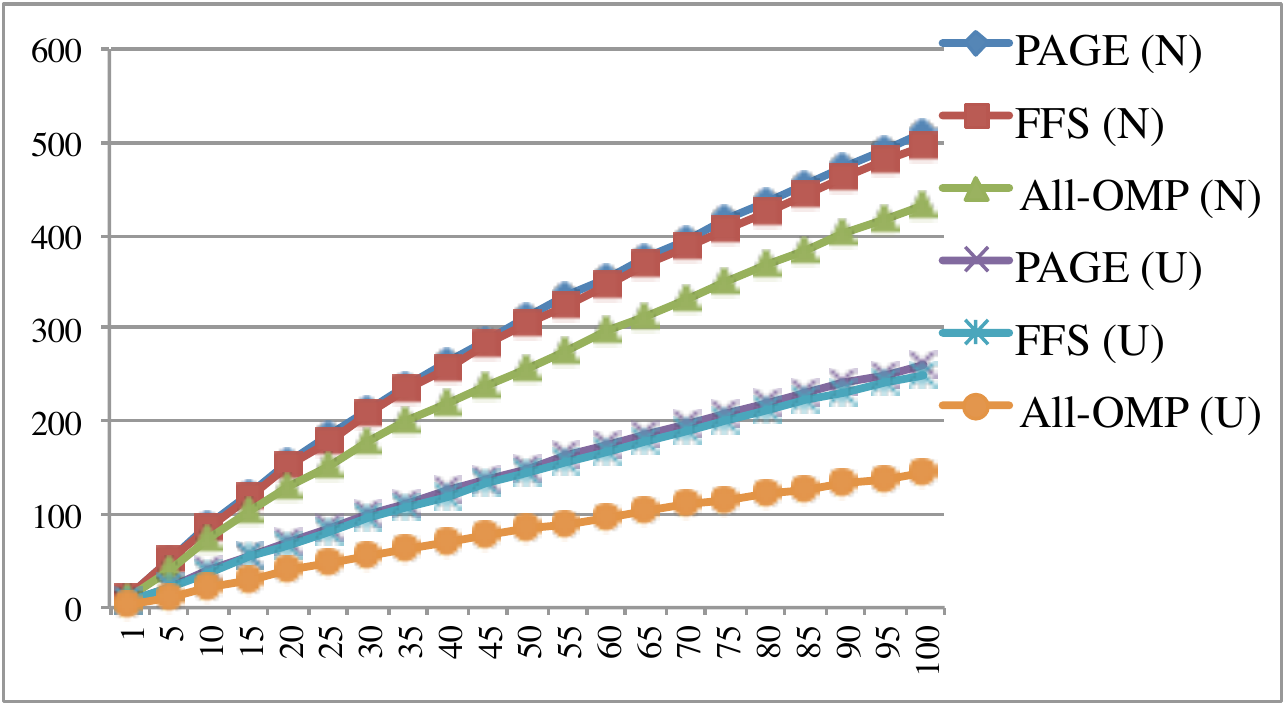} \\
(a) {\small WD with $c_a = 0.1$}  & (b) {\small WD with $c_a = 0.001$} \\
    \includegraphics[width=.45\textwidth]{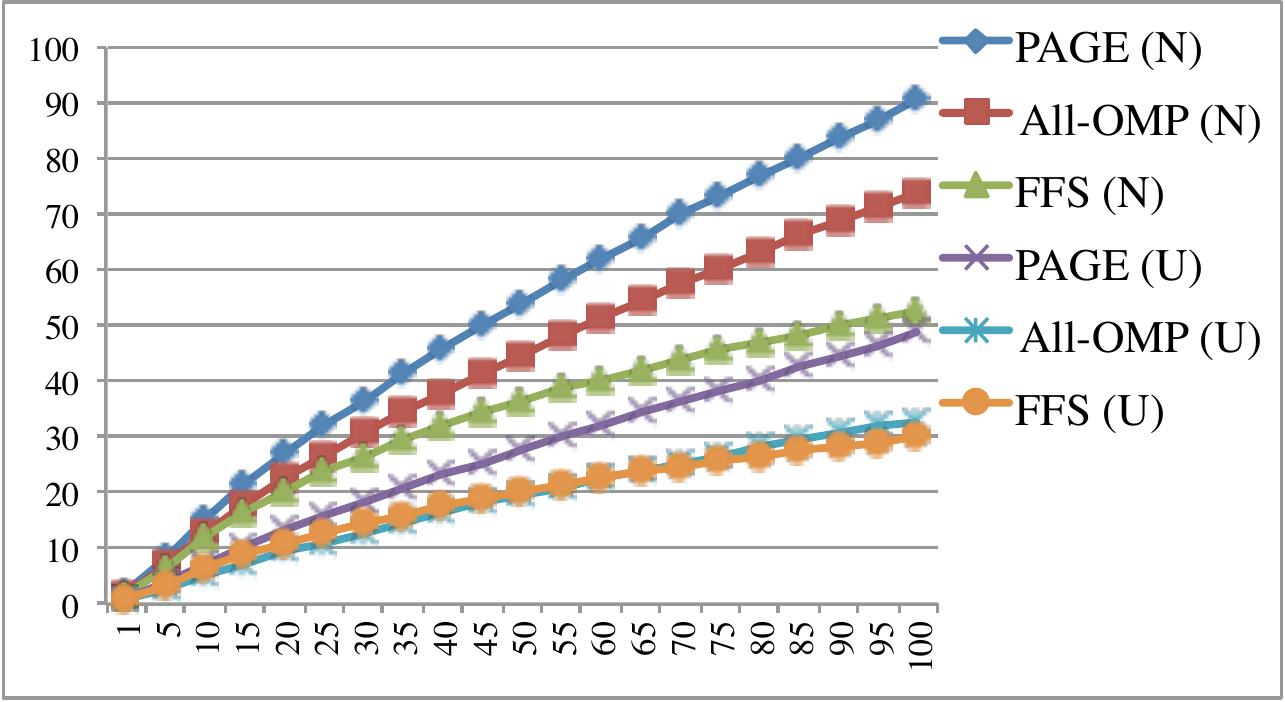}&
    \includegraphics[width=.45\textwidth]{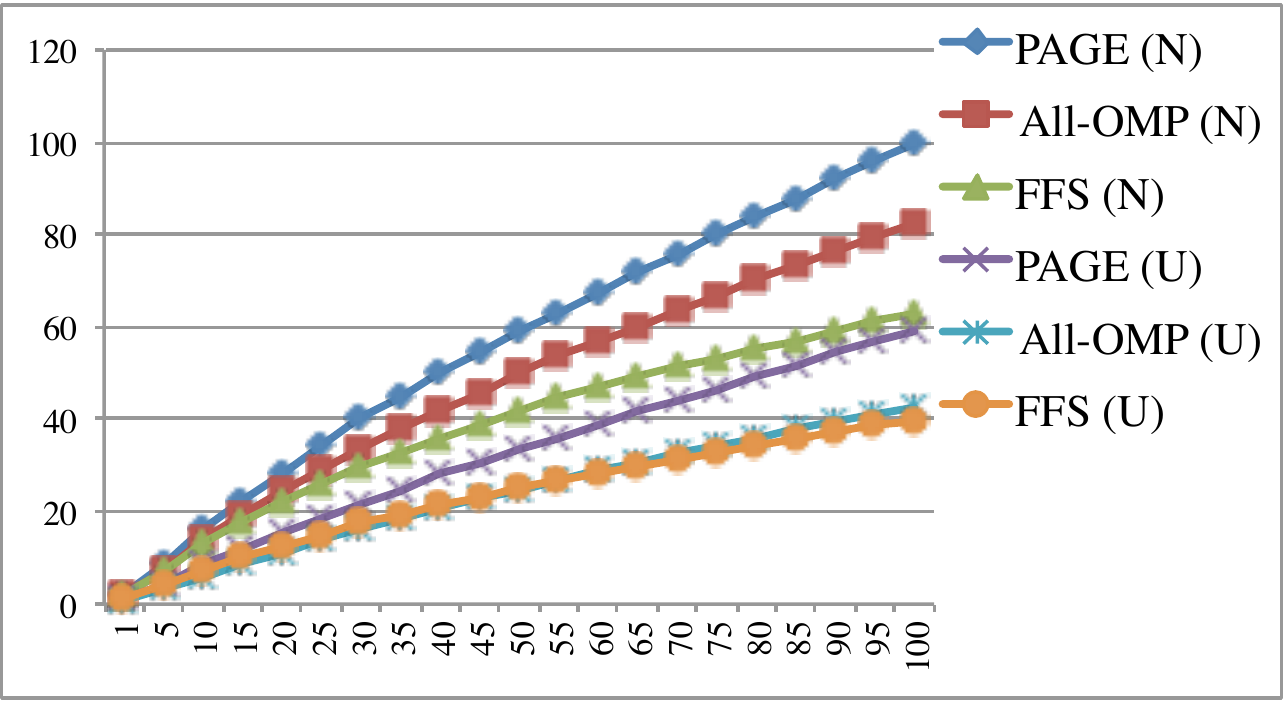}\\
	 (c) {\small TV with $c_a = 0.1$} & (d) {\small TV with $c_a = 0.001$}  \\
\end{tabular}
\caption{Expected profit achieved (Y-axis) on NetHEPT graphs w.r.t.\ $|S|$ (X-axis). (N)/(U) denotes normal/uniform distribution.}
\label{fig:hepSpread}
\end{figure*}

\paragraph{Expected Profit Achieved.}
The quality of outputs (seed sets and price vectors) of \allomp{}, \ffs{}, and \page{}
 	for general \promax{} are evaluated based on the expected profit achieved.
Fig.~\ref{fig:epiSpread}, \ref{fig:flixSpread}, and \ref{fig:hepSpread} illustrate
	the results on Epinions, Flixster, and NetHEPT, respectively.
On each network, both valuation distributions are tested in four settings:
	WD weights with $c_a = 0.1$ and $0.001$; TV weights with $c_a = 0.1$ and $0.001$.
As prices and valuations are in $[0,1]$, we use $0.1$ to simulate
	high acquisition costs and $0.001$ for low costs.
Except for NETHEPT-TV with $c_a = 0.1$ (Fig.~\ref{fig:hepSpread}c)
	and $0.001$ (Fig.~\ref{fig:hepSpread}d), \ffs{}  is better than \allomp{};
	this indicates that only in NetHEPT-TV,
	influence is low enough so that giving free samples blindly to all seeds do impair profits.

In all test cases, 
\page\ performed consistently better than \ffs{} and \allomp{}.
The margin between \page{} and \ffs{} is higher in TV graphs (by, e.g., $15\%$ on Epinions-TV with $\N(0.53,0.14^2)$, $c_a=0.1$) than that in WD graphs (by, e.g., $2.1\%$ on Epinions-WD with $\N(0.53,0.14^2)$, $c_a=0.1$), as higher influence in WD graphs can potentially bring more compensations for profit loss in seeds for \ffs{}.
Also, the expected profit of all algorithms under $\N(0.53,0.14^2)$
	is higher than that under $\U(0,1)$, since adoption probabilities
	under $\N(0.53,0.14^2)$ are higher.

\begin{figure}[h!t!]
\centering
\includegraphics[width=0.55\textwidth]{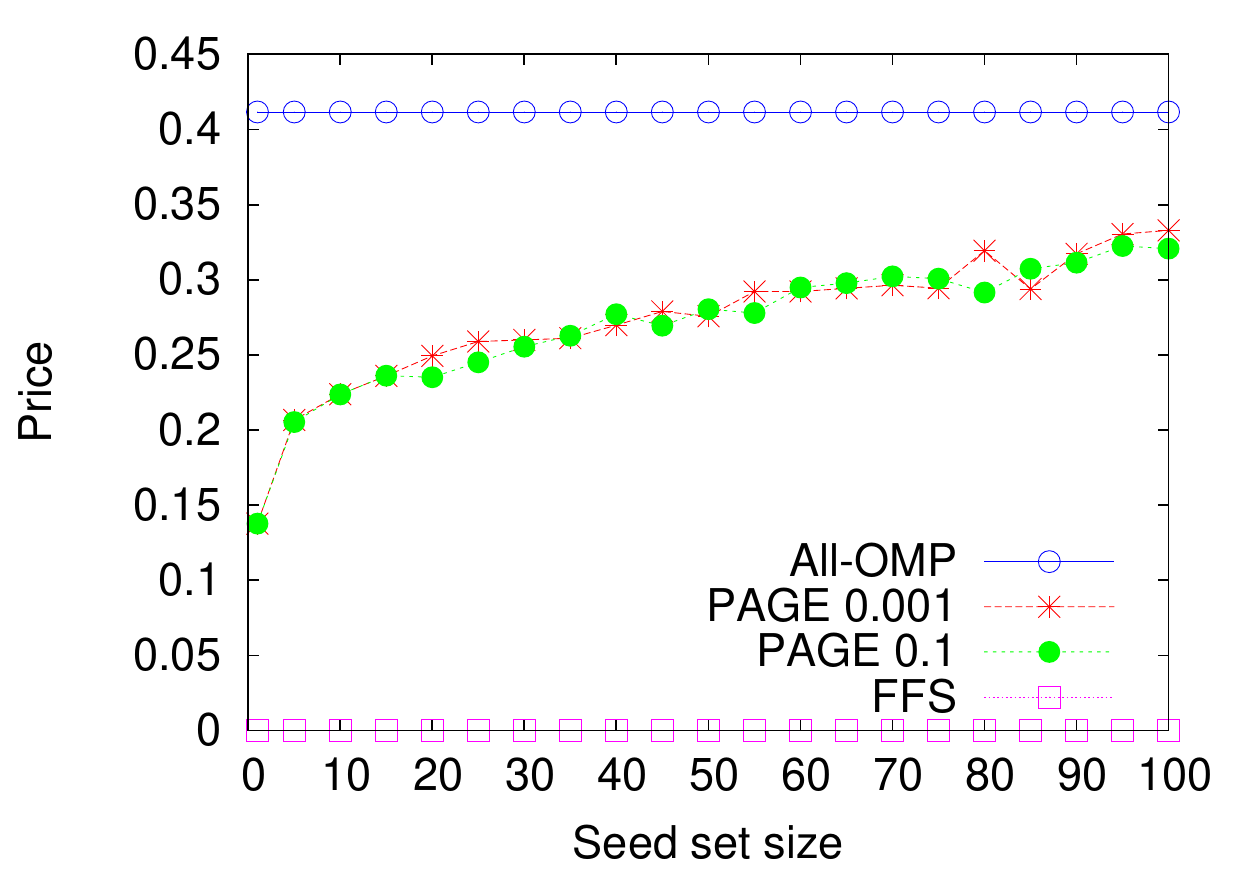}
\caption{Price assigned to seeds (Y-axis) w.r.t.\ $|S|$ (X-axis) on Epinions-TV with $\N(0.53,0.14^2)$.}
\label{fig:page-price}
\end{figure}

\paragraph{Price Assignments.}
For $\N(0.53,0.14^2)$ and $\U(0,1)$, the OMP is $0.41$ and $0.5$, respectively.
Fig.~\ref{fig:page-price} demonstrates the prices offered to each seed by
	\allomp{}, \ffs{}, and \page{} on Epinions-TV with $\N(0.53,0.14^2)$\footnote{Similar results can be seen in other cases, which we omit here.}.
\allomp{} and \ffs{} assigns $0.41$ and $0$ for all seeds, respectively,
For \page{}, as the seed set grows, price tends to increase, reflecting
 	the intuition that discount is proportional to the influence (profit potential) of seeds,
	as they are added in a greedy fashion and those added later 
	have diminishing profit potential.

\begin{table}[h!t!]
	\centering
\begin{tabular}{|c|c|c|c|c|c|c|}
 	 \hline
 	\multirow{2}{*}{\textbf{Algorithm}} &
  	\multicolumn{2}{c|}{\textbf{Epinions-WD}} & 
  	\multicolumn{2}{c|}{\textbf{Flixster-WD}} &
 	\multicolumn{2}{c|}{\textbf{NetHEPT-WD}} \\ \cline{2-7} 
	& $\N$ & $\U$ & $\N$ & $\U$ & $\N$ & $\U$  \\ \hline 
	\allomp{} & 6.7 & 2.3 & 3.0 & 1.0 & 2.6 & 2.2  \\ \hline 
    	\ffs{} & 6.3  & 2.1 & 2.8 & 1.0 & 2.7 & 2  \\ \hline
    	\page{} & 4.8 & 1.3 & 2.3 & 0.5 & 1.0 & 0.9   \\ \hline
\end{tabular}
	\caption{Running time in hours (WD weights, $c_a=0.1$)}
	\label{table:timeWD}
\end{table}

\begin{table}[h!t!]
	\centering
\begin{tabular}{|c|c|c|c|c|c|c|}
 	 \hline
 	\multirow{2}{*}{\textbf{Algorithm}} &
  	\multicolumn{2}{c|}{\textbf{Epinions-TV}} & 
  	\multicolumn{2}{c|}{\textbf{Flixster-TV}} &
 	\multicolumn{2}{c|}{\textbf{NetHEPT-TV}} \\ \cline{2-7} 
	& $\N$ & $\U$ & $\N$ & $\U$ & $\N$ & $\U$  \\ \hline 
	\allomp{} & 5.1 & 2.4 & 1.4 & 1.0 & 2.3 & 2.1  \\ \hline 
    	\ffs{} & 5.5 & 2.5 &  1.5 & 0.8  &  2.2 & 1.8  \\ \hline
    	\page{} & 4.0  & 1.0 & 0.9 & 0.4 & 0.8 & 0.5   \\ \hline
\end{tabular}
	\caption{Running time in hours (TV weights, $c_a=0.1$)}
	\label{table:timeTV}
\end{table}

\paragraph{Running Time.}
Tables~\ref{table:timeWD} and \ref{table:timeTV} present the running time
	of all algorithms on the three networks with WD weights and TV
	weights, respectively\footnote{The results for $c_a=0.001$ are similar, which are omitted here.}.
As adoption probabilities under $\N(0.53,0.14^2)$ are higher,
	all algorithms ran longer with the normal distribution on all graphs.
Similarly, as influence in WD graphs are higher, the running time on them is longer than that on TV graphs.

\allomp{} and \ffs{} have roughly the same running time.
More interestingly, \page{} is faster than both baselines in all cases.
The observation is that in each round of \UG{}, \page{} maximizes
	the marginal profit for each candidate seed in the priority queue
	maintained by CELF.
Thus, heuristically, the lazy-forward procedure in CELF (see \cite{Leskovec07}) has a better
	chance to return the best candidate seed sooner for \page{}.
\allomp{} and \ffs{} also benefit from CELF, but since the marginal profits
	of candidate seeds are often suboptimal, elements in the CELF queue tend to
	be clustered, and thus the lazy-forward is not as effective.
Besides, for \page{} under $\N(0.53,0.14^2)$, the golden section search
	usually converges in less than $40$ iterations with stopping criteria $10^{-8}$ (defined in Sec.~\ref{sec:page}); thus the extra overhead it brings is negligible compared to MC simulations.

To conclude, our empirical results on three real-world datasets with two different valuation distributions demonstrate that the \page{} algorithm consistently outperforms baselines \allomp{} and \ffs{} in both expected profit achieved and running time.
It is also the most robust (against various inputs) among all algorithms.

\section{Conclusions and Discussions}\label{sec:discuss}
In this work, we extend the classical LT model by incorporating prices and valuations to capture
	monetary aspects in product adoption, which we distinguish from social influence.
We study the profit maximization (\promax{}) problem under our proposed LT-V model, and prove \NPhard{ness} and submodularity results.
We propose the \page{} algorithm which dynamically determines the prices for nodes based on their profit potential.
Our experimental results show that \page{} outperforms the baselines in all aspects evaluated.

For future work, first, the added ingredients for LT-V can be used to extend models like IC~\cite{kempe03} and LT-C~\cite{smriti12}.
Second, the current algorithms cannot scale to larger graphs due to expensive MC simulations.
To achieve scalability, we can replace the MC simulations with fast heuristics for the LT model, e.g., LDAG~\cite{ChenYZ10} and SimPath~\cite{simpath}.

Another extension is to consider users' spontaneous interests in product adoption, and incorporate it into the LT-V model for profit maximization.
Due to personal demand, a user may have spontaneous interests in a certain product even when no neighbor in the network has adopted.
To model this, each node $u_i$ is associated with a ``network-less'' probability $\delta_i$~\cite{domingos01}.
An inactive node becomes \emph{influenced} when the sum of $\delta_i$ and the total influence from its \emph{adopting} neighbors are at least $\theta_i$.
A marketing company can thus wait for spontaneous adopters to emerge first and propagate their adoption (for $\ell$ time steps, where $\ell$ is the diameter of $G$), and then deploy a viral marketing campaign to maximize the expected profit.
Our analysis and solution framework (Secs.~\ref{sec:model}, \ref{sec:algo}) can be naturally applied to this setting.

In addition, it is interesting to look into more sophisticated methodologies to acquire knowledge on user valuations, e.g., by leveraging users’ full previous transaction history, as well as look at real datasets besides \texttt{Epinions.com}.


\section*{Acknowledgments}
We thank Wei Chen for helpful comments and discussions on an earlier draft of this paper.
We thank Pei Lee for his help in preparing Epinions price data.
We also thank Allan Borodin, Kevin Leyton-Brown, Min Xie, and Ruben H.\ Zamar for helpful conversations.

This research was supported by a grant from the Business Intelligence Network (BIN) of the Natural Sciences and Engineering Research Council (NSERC) of Canada.

\bibliographystyle{abbrv}
\bibliography{singlebib}

\end{document}